\documentclass[preprint,10pt]{elsarticle}
\usepackage[utf8]{inputenc}
\usepackage{algpseudocode}
\usepackage{algorithm}

\usepackage[margin=2.5cm]{geometry}
\usepackage{graphicx}
\usepackage{amsmath}
\usepackage{amssymb}
\usepackage{amsfonts}
\usepackage{amsthm}
\usepackage{url}
\usepackage[all]{xy}
\usepackage{makeidx}
\usepackage{latexsym}
\usepackage{empheq}
\usepackage{verbatim}
\usepackage{stmaryrd}
\usepackage{enumitem}
\usepackage{xspace}
\usepackage{bussproofs}
\usepackage{boxedminipage}
\usepackage{hyperref}
\usepackage{ccaption}
\usepackage{cancel}
\usepackage{multicol}
\usepackage{xcolor}
\usepackage{proof}
\usepackage{longtable}
\usepackage{tablefootnote}
\usepackage{marginnote}

\newcommand{\jty}{J^{\infty}}

\newcommand{\mpnote}[1]{\textcolor{red}{MP:#1}}

\newcommand{\pcon}{\mathcal{C}}
\newcommand{\npcon}{\cancel{\mathcal{C}}}
\newcommand{\ppcon}{\mathcal{D}}
\newcommand{\nppcon}{\cancel{\mathcal{D}}}

\newcommand{\wt}{{\rhd}}
\newcommand{\bt}{{\blacktriangleright}}
\newcommand{\tw}{{\lhd}}
\newcommand{\tb}{{\blacktriangleleft}}
\usepackage{forest}
\usepackage{cmll}

\usetikzlibrary{patterns}

\def\fCenter{\vdash}

\ifdefined\sep
\else
  \newcommand{\sep}{\ \vert\ }
\fi

\newcommand{\wdia}{\ensuremath{\Diamond}\xspace}

\def\aol{\rule[0.5865ex]{1.38ex}{0.1ex}}

\theoremstyle{plain}
\newtheorem{thm}{Theorem}[section]
\newtheorem{theorem}[thm]{Theorem}

\newtheorem{proposition}[thm]{Proposition}

\newtheorem{lemma}[thm]{Lemma}

\theoremstyle{definition}

\newtheorem{definition}[thm]{Definition}

\newtheorem{example}[thm]{Example}

\newtheorem{remark}[thm]{Remark}

\def\pdra{\mbox{$\,>\mkern-8mu\raisebox{-0.065ex}{\aol}\,$}}

\tikzset{
	treenode/.style = {align=center, inner sep=0pt, text centered},
	Ske/.style = {treenode, ellipse, double, draw=black,
		minimum width=6pt, thick},
	PIA/.style = {treenode, ellipse, black, draw=black,
		minimum width=6pt},
	Crit/.style = {treenode, rectangle, draw=black,
		minimum width=0.5em, minimum height=0.5em}
}

\def\pdla{\mbox{\rotatebox[origin=c]{180}{$\,>\mkern-8mu\raisebox{-0.065ex}{\aol}\,$}}}

\journal{Artificial Intelligence Journal}

\begin{document}
\begin{frontmatter}

\title{Correspondence and Inverse Correspondence for Input/Output Logic and Region-Based Theories of Space }

\author[inst1]{Andrea De Domenico}
\author[inst2]{Ali Farjami}
\author[inst1]{Krishna Manoorkar}
\author[inst1,inst3]{Alessandra Palmigiano}
\author[inst1]{Mattia Panettiere}
\author[inst1,inst4]{Xiaolong Wang*}

\cortext[cor1]{Corresponding author: Xiaolong Wang, Email: 390292381@qq.com }

\affiliation[inst1]{organization={School of Business and Economics, Vrije Universiteit Amsterdam},
            country={
            The Netherlands}}

\affiliation[inst2]{organization={Interdisciplinary Centre for Security Reliability and Trust, University of Luxembourg},
            country={Luxembourg}}

\affiliation[inst3]{organization={Department of Mathematics and Applied Mathematics, University of Johannesburg},
            country={South Africa}}

\affiliation[inst4]{organization={School of Philosophy and Social Development, Shandong University},
            addressline={Jinan},
            country={China}
            }

\begin{abstract}
We further develop the algebraic approach to input/output logic initiated in \cite{wollic22}, where subordination algebras and a family of their generalizations were proposed as a semantic environment of various input/output logics. In particular: we extend  the  modal characterizations of a finite number of well known conditions on normative and permission systems, as well as on subordination, precontact, and dual precontact algebras developed in \cite{de2024obligations}, to those corresponding to the infinite class of {\em clopen-analytic inequalities} in a modal language consisting  both of positive and of negative unary modal operators; we characterize the syntactic shape of first-order conditions on algebras endowed with subordination, precontact, and dual precontact relations which guarantees these conditions to be the first-order correspondents of axioms in the modal language above; we introduce algorithms  for computing the first-order correspondents of modal axioms on algebras endowed with subordination, precontact, and dual precontact relations, and conversely, for computing the modal axioms of which the conditions satisfying the suitable syntactic shape are the first-order correspondents; finally, we extend Celani's dual characterization results between subordination lattices and subordination spaces to a wider environment which also encompasses precontact and dual precontact relations, and relative to an infinite class of first order conditions relating subordination, precontact and dual precontact relations on distributive lattices. The modal characterizations established in the present paper  pave the way to establishing faithful embeddings for infinite classes of input/output logics, and hence to their implementation  in LogiKEy, Isabelle/HOL, Lean, or other interactive systems.
\smallskip


\end{abstract}

\end{frontmatter}

\section{Introduction}

Input/output logic \cite{Makinson00} has been introduced  as a formal framework for studying the interaction between logical inferences and other agency-related relations between formulas, encoding e.g.~conditional obligations, goals, ideals, preferences, actions, and beliefs, 
and has been applied not only in the context of the formalization of normative systems in philosophical logic and AI \cite{Parenthand}, but also in connection and combination with a very diverse range of issues, also relevant to present-day AI, spanning from formal argumentation \cite{berkel2022reasoning, chen2024bisimulation} to causal reasoning \cite{bochman2004causal}, and non-monotonic reasoning \cite{Makinson01}. 
In most of these research lines, and in fact, since its inception, the issue of establishing connections between the various forms of outputs associated with normative systems and  {\em modal logic} formalisms, and especially with deontic logic, has been deemed  worth being pursued, since the existence of this connection would grant access to a number of techniques and results, both model-theoretic \cite{strasser2016adaptive} and  proof-theoretic (e.g.~tableaux, decidability results, proof-search \cite{Lellmann2021}), which can shed light on the computational behaviour and performances of input/output logic. In particular, modal characterizations of input/output logic would allow for the powerful and flexible LogiKEy methodology \cite{Christoph2020} to be available  for the development and experimentation with ethical reasoners and normative theories, based on a semantic embedding of deontic logic within (decidable fragments of) higher-order logic (HOL). 
Modal characterizations of two of the four original   types of outputs   have been achieved by syntactic means  since the inception of the input/output framework \cite[Sections 4.3 and 5.5]{Makinson00}, and were used to obtain a faithful semantic embedding of the corresponding input/output logics into classical higher order logic, which allowed for their implementation into Isabelle/HOL \cite{benzmuller2019}.  More recently, the modal embeddings of the original outputs, together with those obtained via semantic interpretation  in the context of the neighbouring setting of causal reasoning   \cite{bochman2004causal},  have been used to design proof-search oriented proof-calculi  for various input/output logics and causal reasoning \cite{Agata2023}.

The present paper  is the prosecution of a line of investigation,  initiated in \cite{wollic22} and continued in \cite{de2024obligations, de2024obligations2}, aimed at establishing systematic translations and algorithmically generating modal characterization of {\em classes}  of outputs on {\em classes} of propositional logic settings. The methodology adopted in this line of research pivots on the recent introduction of {\em subordination algebras} and related structures  as a semantic environment for input/output logic \cite{wollic22}. 

Subordination algebras \cite{bezhanishvili2017irreducible} and related structures such as pre-contact algebras \cite{dimov2005topological} and quasi-modal algebras \cite{celani2001quasi,celani2016precontact} have been introduced in the context of a point-free approach to the region-based theories of discrete spaces, and have mainly been investigated with algebraic and duality-theoretic methods: for instance, the link established between subordination algebras and certain Boolean algebras with operators \cite{jonsson1951boolean} has made it possible to endow  various modal languages with a semantics  based on subordination algebras, and use these languages to modally characterize their properties. In particular, Sahlqvist-type canonicity for modal axioms on subordination algebras has been studied in \cite{de2020subordination} using topological techniques. Very similar structures  to subordination algebras have  been introduced by Lindahl and Odelstadt \cite{lindahl2013theory} in the context of an algebraic approach for the representation of normative systems; however, the explicit connection with subordination algebras allows access to a  systematic link between topological and algebraic insights in establishing general and uniform modal characterization results for classes of input/output logics. 
Specifically relevant to the present paper, in \cite{de2021slanted}, using algebraic techniques, the canonicity result  of \cite{de2020subordination} was strengthened and captured within the more general notion of canonicity in the context of {\em slanted algebras}, which was established using the tools of {\em unified correspondence theory} \cite{conradie2014unified,conradie2019algorithmic,conradie2020constructive}. Slanted algebras are based on general lattices, and encompass variations and generalizations of subordination algebras such as those very recently introduced by Celani in \cite{celani2020subordination}, which are based on distributive lattices, and for which Celani develops  duality-theoretic and correspondence-theoretic results for  finitely many axioms.

In the present paper, we directly build on \cite{de2024obligations2}, and generalize the modal characterizations of finitely many  conditions on normative and permission systems, as well as on subordination algebras and related structures developed in \cite{de2024obligations2}, to those corresponding to the infinite class of {\em clopen-analytic inequalities} (cf.~Definition \ref{def:clopen_analytic}).  Conversely, we  introduce {\em Kracht formulas} (cf.~Definition \ref{def:inverse_shape})  as the class of those first order conditions  on algebras with subordination, precontact, and dual precontact relations  which correspond to clopen-analytic  axioms; finally, as an application of these results, we generalize Celani's dual characterization results  between subordination lattices and subordination spaces  (cf.~\cite[Theorem 5.7]{celani2020subordination}) to the environment of spd-algebras based on distributive lattices and the infinite  class of  inductive Kracht formulas (cf.~Definition \ref{def:kracht_analytic_inductive}).

The modal characterizations established in the present paper  pave the way to establishing faithful embeddings for infinite classes of input/output logics, and hence to their implementation  in LogiKEy, Isabelle/HOL, Lean, or other interactive systems.  
Moreover, these characterizations modularly apply to a wide class  of algebraic settings,  corresponding to varieties of normal lattice expansions \cite{gehrke2001bounded} and their associated logics \cite{conradie2019algorithmic}. Hence, the present contributions further develop and strengthen the  principled generalization, initiated in \cite{de2024obligations}, of the theory of normative and permission systems from its original setting  (namely, classical propositional logic) to the very wide class of {\em selfextensional logics}. This class is well known and very well studied in abstract algebraic logic \cite{font2003survey, jansana2006referential,jansana2006selfextensional}, and includes all the best known and most used logics for computer science and AI, such as modal, intuitionistic, bi-intuitionistic and substructural logics, and logics with weaker form of negation, or no negation althogether, such as those pertaining to the family of  equilibrium logic  \cite{pearce2006equilibrium,cabalar2023deontic}.
Related to this observation, the logical signatures  considered in the present paper account both for obligations and various forms of permissions considered as primitive notions, as well as their {\em interaction};\footnote{Likewise, in neighbouring settings such as Bochmann's framework for causal reasoning \cite{bochman2004causal,bochman2003logic}, Kracht's formulas and analytic axioms of  suitable signatures can be used to represent relevant forms of interaction between the `is the cause of' and the `not prevents' relationships.}

\paragraph{Structure of the paper} In Section \ref{sec:prelim}, we collect basic definitions and facts about canonical extensions, algebras endowed with subordination, precontact, and dual precontact relations (referred to as {\em spd-algebras}, cf.~Definition \ref{def:spd-algebra}), their associated slanted algebras and canonical extensions, the first-order logic associated with spd-algebras, the propositional logics associated with the slanted algebras arising from spd-algebras, and the (clopen-)analytic axioms in the languages of these logics. In Section \ref{sec:correspondence}, we show that every clopen-analytic axiom corresponds to a first order condition in the language of spd-algebras. In Section \ref{sec:examples}, we illustrate the application of this procedure on a concrete example. In Section \ref{sec:kracht_formulas}, we introduce {\em Kracht formulas} (cf.~Definition \ref{def:kracht_analytic_inductive}) as the  first-order formulas in the languages of spd-algebras the syntactic shape of which guarantees them to be the first-order correspondents of some propositional axioms, and in Section \ref{sec:inverse_correspondence_procedure}, we present an algorithm for computing their associated  propositional axioms. In Section \ref{sec: applications}, we use the results of the previous section to  extend Celani's dual characterization  results \cite[Theorem 5.7]{celani2020subordination} and previous results  \cite[Propositions 6.11 and 6.12]{de2024obligations2}  to   inductive Kracht formulas. We conclude in Section \ref{sec:conclusions}.

\section{Preliminaries}
\label{sec:prelim}


 \subsection{Canonical extensions of bounded lattices}
\label{ssec:canext lattices}

In the present subsection, we adapt material from  \cite[Sections 2.2 and 3.1]{de2020subordination},\cite[Section 2]{DuGePa05}. In what follows,  when we say `lattice', we mean `bounded lattice'. We will make use of the following notation throughout the present paper: an {\em order-type} over $n\in \mathbb{N}$ is an $n$-tuple $\epsilon\in \{1, \partial\}^n$.  For every order type $\epsilon$, we denote its {\em opposite} order type by $\epsilon^\partial$, that is, $\epsilon^\partial_i = 1$ iff $\epsilon_i=\partial$ for every $1 \leq i \leq n$. For any lattice $A$, we let $A^1: = A$ and $A^\partial$ be the dual lattice of $A$, that is, the lattice associated with the converse partial order of $A$. For any order type $\varepsilon$ over $n$, we let $A^\varepsilon: = \Pi_{i = 1}^n A^{\varepsilon_i}$. 
\begin{definition}
Let $A$ be a sublattice 
of a complete lattice $A'$.
\begin{enumerate}
\item An element $k\in A'$ is {\em closed} if $k = \bigwedge F$ for some non-empty 
$F\subseteq A$; an element $o\in A'$ is {\em open} if $o = \bigvee I$ for some non-empty 
$I\subseteq A$;
\item  $A$ is {\em dense} in $A'$ if every element of $A'$ is both  the join of closed elements and  the meet of open elements of $A'$.
\item $A$ is {\em compact} in $A'$ if, for all nonempty $F, I\subseteq A$, 
if $\bigwedge F\leq \bigvee I$ then $a\leq b$ for some $a\in  F$ and $b\in I$. 
\item The {\em canonical extension} of a lattice $A$ is a complete lattice $A^\delta$ containing $A$
as a dense and compact sublattice.
\end{enumerate}
\end{definition}
The canonical extension $A^\delta$ of any lattice $A$ always exists\footnote{For instance, the canonical extension of a distributive lattice  $A$ is (isomorphic to) the  complete and completely distributive lattice $\mathcal{P}^{\uparrow}(Prime(A))$ of the upward closed subsets of the poset $Prime(A)$  of the prime filters of $A$ ordered by  inclusion.} and is  unique up to an isomorphism fixing $A$ (cf.\ \cite[Propositions 2.6 and 2.7]{DuGePa05}).
%
We let $K(A^\delta)$ (resp.~$O(A^\delta)$) denote the set of the closed (resp.~open) elements    of $A^\delta$. It is easy to see that $A = K(A^\delta)\cap O(A^\delta)$, which is why the elements of $A$ are referred to as the {\em clopen} elements of $A^\delta$. The following proposition collects well known facts which we will use in the remainder of the paper. In particular, item  (iv) is a variant of \cite[Lemma 3.2]{gehrke2001bounded}.

 \begin{proposition} \label{prop:background can-ext} (cf.~\cite[Proposition 2.6]{de2024obligations2})
 For any lattice $A$,
  all $k_1, k_2\in K(A^\delta)$, $o_1, o_2\in O(A^\delta)$, and  $u_1, u_2\in A^\delta$,
 \begin{enumerate}[label=(\roman*)]
     \item $k_1\leq k_2$ iff $k_2\leq b$ implies $k_1\leq b$ for all $b\in A$.
     \item $o_1\leq o_2$ iff $b\leq o_1$ implies $b\leq o_2$ for all $b\in A$.
     \item $u_1\leq u_2$ iff  $k\leq u_1$ implies $k\leq u_2$ for all $k\in K(A^\delta)$, iff $u_2\leq o$ implies $u_1\leq o$ for all $o\in O(A^\delta)$.
     \item 
     $k_1\vee k_2\in K(A^\delta)$ and  $o_1\wedge o_2\in O(A^\delta)$.
  \end{enumerate}
 \end{proposition}

\begin{proposition}[cf.~\cite{de2024obligations2}, Proposition 2.7]
    \label{prop: compactness and existential ackermann}
 For any  lattice $A$, 
\begin{enumerate}
   \item for any $b\in A$,
   $k_1, k_2 \in K(A^\delta)$, and $o \in O(A^\delta)$, 
    \begin{enumerate}[label=(\roman*)]
        \item $ k_1\wedge k_2 \leq b$ implies $a_1\wedge a_2\leq b$ for some $a_1, a_2 \in A$ s.t.~$ k_i \leq a_i$;
         \item $ k_1\wedge k_2 \leq o$ implies $a_1\wedge a_2\leq b$ for some $a_1, a_2 ,b \in A$ s.t.~$ k_i \leq a_i$ and $b\leq o$;
             \item 
             $\bigwedge K\in K(A^\delta)$ for every  $K\subseteq K(A^\delta)$.
    \end{enumerate}
    \item for any $a\in A$,
   $o_1, o_2 \in O(A^\delta)$, and $k \in K(A^\delta)$, 
    \begin{enumerate}[label=(\roman*)]
        \item $a\leq o_1\vee o_2 $ implies $a \leq b_1\vee b_2$ for some $b_1, b_2 \in A$ s.t.~$  b_i\leq o_i$;
         \item $ k \leq o_1\vee o_2 $ implies $a \leq b_1\vee b_2$ for some $a, b_1, b_2 \in A$ s.t.~$ b_i \leq o_i$ and $ k \leq a$.
         \item 
         $\bigvee O\in O(A^\delta)$ for every  $O\subseteq O(A^\delta)$.
    \end{enumerate}

    \end{enumerate}
\end{proposition}

 \subsection{Subordination and (dual) precontact  relations on bounded lattices}
\label{ssec:subordination precontact}

Let $A$ be a bounded lattice. A binary relation $\pcon$ on the domain of $A$ is a {\em precontact relation} \cite{dimov2005topological, duntsch2007region} if, for all $a, b, c\in A$,
 \begin{enumerate}
\item[(C1)] $a \pcon b$ implies $a, b \neq\bot$;
\item[(C2)] $a \pcon (b \vee c)$ iff  $a \pcon b$ or $a \pcon c$;
\item[(C3)] $(a \vee b) \pcon  c$ iff  $a \pcon c$ or $b \pcon c$.
\end{enumerate}
A binary relation $\prec$ on the domain of $A$ is a {\em subordination relation} \cite{dimov2005topological, duntsch2007region} if, for all $a, b,  c, d\in A$,
 \begin{enumerate}
\item[(S1)] $\bot\prec\bot$ and $\top\prec\top$;
\item[(S2)] if $a \prec b$ and $a\prec c$ then $a \prec b \wedge c$;
\item[(S3)]  if $a \prec c$ and $b\prec c$ then $a\vee b \prec  c$;
\item[(S4)] if $a \leq b \prec c \leq d$ then $a \prec d$.
\end{enumerate}

A binary relation $\ppcon$ on the domain of $A$ is a {\em dual precontact relation} \cite[Section 2.3]{de2024obligations2} if, for all $a, b, c\in A$,
 \begin{enumerate}
\item[(D1)] $a \ppcon b$ implies $a, b \neq\top$;
\item[(D2)] $a \ppcon (b \wedge c)$ iff  $a \ppcon b$ or $a \ppcon c$;
\item[(D3)] $(a \wedge b) \ppcon  c$ iff  $a \ppcon c$ or $b \ppcon c$.
\end{enumerate}
The following definition introduces the first of the two main semantic environments of the present paper.
  \begin{definition}[\cite{de2024obligations2}, Footnote 16]
  \label{def:spd-algebra}
      An {\em spd-algebra} is a structure $\mathbb{H} = (A, \mathsf{S}, \mathsf{C}, \mathsf{D})$ s.t.~$A$ is an  LE-algebra\footnote{\label{ftn:LE-algebra}An {\em LE-algebra} (cf.~\cite[Definition 1.3]{conradie2019algorithmic}) is a structure $A = (L, \mathcal{F}_A, \mathcal{G}_A)$ s.t.~$L$ is a bounded lattice, and $\mathcal{F}_A$ and $ \mathcal{G}_A$ are finite sets of operations on $L$, s.t.~each $f\in\mathcal{F}_A$ (resp.~$g\in\mathcal{G}_A$) is endowed with an arity $n_f$ (resp.~$n_g$) and an order-type $\epsilon_f: \{1,\ldots ,n_f\}\to \{1, \partial\}$ (resp.~$\epsilon_g:\{1,\ldots , n_g\}\to \{1, \partial\}$), and is finitely join-preserving (resp.~finitely meet-preserving) if $\epsilon_f(i) = 1$ (resp.~$\epsilon_g(i) = 1$), and is  finitely meet-reversing (resp.~finitely join-reversing) if $\epsilon_f(i) = \partial$ (resp.~$\epsilon_g(i) = \partial$). E.g., a Boolean algebra is a distributive LE-algebra with $\mathcal{F}_A = \mathcal{G}_A = \{\neg\}$ with $n_{\neg} = 1$ and $\varepsilon_{\neg} = \partial$; a bi-Heyting algebra is a distributive LE-algebra with $\mathcal{F}_A = \{\pdla\}$, $ \mathcal{G}_A = \{\rightarrow\}$ with $n_{{\tiny\pdla}} = n_{\rightarrow} = 2$, $\varepsilon_{\rightarrow}(1) = \varepsilon_{{\tiny{\pdla}}} (2) = \partial$ and $\varepsilon_{\rightarrow}(2) = \varepsilon_{{\tiny \pdla}} (1) = 1$; a Heyting (resp.~co-Heyting) algebra is an LE-algebra of the $\pdla$-free (resp.~$\rightarrow$-free) fragment of the bi-Heyting algebra similarity type.}, and $\mathsf{S}$, $\mathsf{C}$, and $\mathsf{D}$ are finite families of subordination, precontact, and dual precontact relations on $A$, respectively. The {\em spd-type} of $\mathbb{H}$ is the triple of natural numbers $(|\mathsf{S}|, |\mathsf{C}|, |\mathsf{D}|)$.
  \end{definition}
  Spd-algebras form an environment in which normative, permission, and dual permission systems can be represented as subordination, precontact, and dual precontact relations, respectively. While these notions  are interdefinable in the setting of Boolean algebras\footnote{For instance,   any subordination relation $\prec$ on a Boolean algebra $A$ gives rise to a precontact relation $\pcon_{\prec} $ defined as $a\pcon_{\prec} b$  iff $a\not\prec\neg b$, and a dual precontact relation  $\ppcon_{\prec} $ defined as $a\ppcon_{\prec} b$  iff $\neg a\not\prec b$; for an expanded discussion, see \cite[Section 2.3]{de2024obligations2}.}, in the present, more general setting they all need to be regarded as primitive, and their inter-relations to be described via suitable axioms.
  In what follows, we will  consider spd-algebras based on distributive lattices, Boolean, Heyting, co-Heyting, and bi-Heyting algebras, De Morgan algebras and normal modal expansions thereof.

  We will refer to the operations on the domain algebra of a given spd-algebra as {\em standard},  {\em non-slanted}, or {\em clopen} (this terminology will be motivated in Section \ref{ssec:slanted spd}).
  \begin{definition}\label{def:fo-language-spd} For any spd-type $(n_\mathsf{S}, n_\mathsf{C}, n_\mathsf{D})$, and any variety of algebras $\mathsf{K}$,
      the first order language $\mathcal{L}^{\mathrm{FO}}_{spd} = \mathcal{L}^{\mathrm{FO}}_{spd} (|\mathsf{S}|, |\mathsf{C}|, |\mathsf{D}|, \mathsf{K})$ is determined by the following non-logical symbols:
      \begin{enumerate}
          \item function symbols (including constants) parametric in the signature of $\mathsf{K}$, each with corresponding arity;
          \item binary relation symbols parametric in the given spd-type.
      \end{enumerate}
  \end{definition}
  In the context of the theory developed in the present paper, the first order language above serves as the `frame correspondence' language of the modal language discussed in Section \ref{ssec:LE-logic}.

\subsection{Slanted algebras associated with spd-algebras}
\label{ssec:slanted spd}
In the previous section, we introduced the  `relational semantics' of  the modal language discussed in Section \ref{ssec:LE-logic},  while the structures introduced in the present section will serve as its algebraic semantics.
\begin{definition}
 For any  spd-algebra $\mathbb{H} = (A, \mathsf{S}, \mathsf{C}, \mathsf{D})$, the {\em slanted algebra} associated with $\mathbb{H}$ is    $\mathbb{H}^\ast = (A, \mathcal{F}^{\mathbb{H}^\ast}, \mathcal{G}^{\mathbb{H}^\ast})$ s.t.~$\mathcal{F}^{\mathbb{H}^\ast}: = \mathcal{F}_{A}\cup\mathcal{F}_{\mathsf{S}}\cup \mathcal{F}_{\mathsf{D}}$ and $\mathcal{G}^{\mathbb{H}^\ast}: = \mathcal{G}_{A}\cup\mathcal{G}_{\mathsf{S}}\cup \mathcal{G}_{\mathsf{C}}$, where $\mathcal{L}(\mathcal{F}_{A},\mathcal{G}_{A})$ is the (possibly empty) signature of the LE-algebra $A$, and moreover,
 \begin{center}
 \begin{tabular}{ll}
 $\mathcal{F}_{\mathsf{S}}: = \{\Diamond_{\prec}\mid {\prec}\in \mathsf{S}\}$ &
 $ \mathcal{G}_{\mathsf{S}}: = \{\blacksquare_{\prec}\mid {\prec}\in \mathsf{S}\}$\\
 $\mathcal{F}_{\mathsf{D}} : = \{\tw_{\mathcal{D}}, \tb_{\mathcal{D}}\mid \mathcal{D}\in \mathsf{D}\}$
 &  $\mathcal{G}_{\mathsf{C}} : = \{\wt_{\mathcal{C}}, \bt_{\mathcal{C}}\mid \mathcal{C}\in \mathsf{C}\}$
 \end{tabular}
 \end{center}
 and the elements in the sets displayed above are unary maps $A\to A^\delta$, defined by the assignments indicated in the table below, for every ${\prec}\in\mathsf{S}$, every $\pcon\in \mathsf{C}$, and every $\ppcon\in \mathsf{D}$.
 \begin{center}
     \begin{tabular}{|c|l||c|l|}
     \hline
     $\Diamond_{\prec} $ & $a\mapsto\bigwedge {\prec}[a]: = \bigwedge \{b\mid a\prec b\}\in K(A^\delta)$
        &
     $\blacksquare_{\prec} $ & $a\mapsto\bigvee {\prec}^{-1}[a]: = \bigvee \{b\mid b\prec a\}\in O(A^\delta)$\\
      \hline
       $\wt_{\pcon}$   & $a\mapsto\bigvee ({\pcon}[a])^c: = \bigvee \{b\mid a\npcon b\}\in O(A^\delta)$ &
        $\bt_{\pcon}$  & $a\mapsto\bigvee ({\pcon}^{-1}[a])^c: = \bigvee \{b\mid b\npcon a\}\in O(A^\delta)$\\
         \hline
         $\tw_{\ppcon}$   & $a\mapsto\bigwedge ({\ppcon}[a])^c: = \bigwedge \{b\mid a\nppcon b\}\in K(A^\delta)$ &
        $\tb_{\ppcon}$  & $a\mapsto\bigwedge ({\ppcon}^{-1}[a])^c: = \bigwedge \{b\mid b\nppcon a\}\in K(A^\delta)$\\
         \hline
     \end{tabular}
 \end{center}
\end{definition}
In what follows, we will typically omit the subscripts, and write e.g.~$\tw$ for any $\tw_{\ppcon}$ s.t.~$\ppcon\in\mathsf{D}$. Slanted algebras were defined in \cite{de2021slanted} for arbitrary LE-signatures. However, in what follows, when referring to slanted algebras, we will understand those with a similarity type compatible with those of the definition above. From the definitions above, it immediately follows that, for every spd-algebra $\mathbb{H}$ and all $a, b, \in A$:
\begin{center}
\begin{tabular}{c c ccc}
   $\Diamond a\leq b$  & iff & $a\prec b$ & iff & $a\leq \blacksquare b$ \\
   $b \leq \wt a$  & iff & $a\npcon b$ & iff & $a\leq \bt b$ \\
     $\tw b \leq  a$  & iff & $a\nppcon b$ & iff & $\tb a\leq  b$ \\
\end{tabular}
\end{center}
\begin{definition}
\label{def: sigma and pi extensions of slanted}
(cf.~\cite[Definition 1.10]{de2020subordination})
For any spd-algebra $\mathbb{H} = (A, \mathsf{S}, \mathsf{C}, \mathsf{D})$, the {\em canonical extension} of the slanted algebra $\mathbb{H}^\ast$ is the (standard!) LE-algebra $\mathbb{H}^\delta: = (A^\delta, \mathcal{F}^{\mathbb{H}^\delta}, \mathcal{G}^{\mathbb{H}^\delta})$ s.t.~$\mathcal{F}^{\mathbb{H}^\delta}: = \mathcal{F}_A\cup \mathcal{F}_{\mathsf{S}}\cup \mathcal{F}_{\mathsf{D}}$ and $\mathcal{G}^{\mathbb{H}^\delta}: = \mathcal{G}_A\cup \mathcal{G}_{\mathsf{S}}\cup \mathcal{G}_{\mathsf{C}}$, and
\begin{center}
 \begin{tabular}{ll}
 $\mathcal{F}_{A}: = \{f^\sigma\mid f\in \mathcal{F}_A\}$ &
 $ \mathcal{G}_{\mathsf{S}}: = \{g^\pi\mid g\in \mathcal{G}_A\}$\\
 $\mathcal{F}_{\mathsf{S}}: = \{\Diamond_{\prec}^\sigma\mid {\prec}\in \mathsf{S}\}$ &
 $ \mathcal{G}_{\mathsf{S}}: = \{\blacksquare_{\prec}^\pi\mid {\prec}\in \mathsf{S}\}$\\
 $\mathcal{F}_{\mathsf{D}} : = \{\tw_{\mathcal{D}}^\sigma, \tb_{\mathcal{D}}^\sigma\mid \mathcal{D}\in \mathsf{D}\}$
 &  $\mathcal{G}_{\mathsf{C}} : = \{\wt_{\mathcal{C}}^\pi, \bt_{\mathcal{C}}^\pi\mid \mathcal{C}\in \mathsf{C}\}$
 \end{tabular}
 \end{center}
where (omitting subscripts) $f^\sigma: (A^\delta)^{\epsilon_f}\to A^\delta$, $g^\pi: (A^\delta)^{\epsilon_g}\to A^\delta$ are defined as in \cite[Definition 1.9]{conradie2019algorithmic}, and $\Diamond^\sigma,  \blacksquare^\pi, \wt^\pi, \bt^\pi, \tw^\sigma, \tb^\sigma: A^\delta\to A^\delta$ are defined as follows:  for every $k\in K(A^\delta)$, $o\in O(A^\delta)$ and $u\in A^\delta$,
\[\Diamond^\sigma k:= \bigwedge\{ \Diamond a\mid a\in A\mbox{ and } k\leq a\}\quad \Diamond^\sigma u:= \bigvee\{ \Diamond^\sigma k\mid k\in K(A^\delta)\mbox{ and } k\leq u\}\]
\[\blacksquare^\pi o:= \bigvee\{ \blacksquare a\mid a\in A\mbox{ and } a\leq o\},\quad \blacksquare^\pi u:= \bigwedge\{ \blacksquare^\pi o\mid o\in O(A^\delta)\mbox{ and } u\leq o\}\]
\[\wt^\pi k:= \bigvee\{ \wt a\mid a\in A\mbox{ and } k\leq a\},\quad \wt^\pi u:= \bigwedge\{ \wt^\pi k\mid k\in K(A^\delta)\mbox{ and } k\leq u\}\]
\[\bt^\pi k:= \bigvee\{ \bt a\mid a\in A\mbox{ and } k\leq a\},\quad \bt^\pi u:= \bigwedge\{ \bt^\pi k\mid k\in K(A^\delta)\mbox{ and } k\leq u\}.\]
\[\tw^\sigma o:= \bigwedge\{ \tw a\mid a\in A\mbox{ and }  a \leq o \},\quad \tw^\sigma u:= \bigvee\{ \tw^\sigma o\mid o\in O(A^\delta)\mbox{ and }  u\leq o \}\]
\[\tb^\sigma o:= \bigwedge\{ \tb a\mid a\in A\mbox{ and } a\leq o\},\quad \tb^\sigma u:= \bigvee\{ \tb^\sigma o\mid o\in O(A^\delta)\mbox{ and } u\leq o\}.\]
\end{definition}
In what follows, $\overline a$ and $\overline b$ denote vectors of elements in $A$, and $\overline k$ and $\overline o$ denote  vectors of elements in $K(A^\delta)$ and $O(A^\delta)$, respectively. Moreover, for any $f\in \mathcal{F}$ (resp.~$g\in \mathcal{G}$) which is positive in each $x$-coordinate and negative in each $y$-coordinate (each coordinate being indicated by a placeholder variable, cf.~Section \ref{ssec:analytic-LE-axioms}), we write $f^\sigma[\overline{k}/!\overline x, \overline{o}/!\overline y]$ (resp.~$g^\pi[\overline{o}/!\overline x, \overline{k}/!\overline y ]$) to indicate that the positive coordinates of $f$ (resp.~$g$) have been instantiated with closed (resp.~open) elements, and the negative ones with open (resp.~closed) ones.
\begin{lemma}[cf.~\cite{de2021slanted} Lemma 3.5]
\label{lemma: distribution properties of sigma pi}
   For every spd-algebra $\mathbb{H}$, any  $f\in \mathcal{F}^{\mathbb{H}^\ast}$, and
any  $g\in \mathcal{G}^{\mathbb{H}^\ast}$ which are positive in each $x$-coordinate and negative in each $y$-coordinate,
\begin{enumerate}
    \item  $f^\sigma$ is completely join-preserving in each positive coordinate, and  completely meet-reversing in each negative coordinate. Moreover, $f^\sigma[\overline{k}/!\overline x, \overline{o}/!\overline y]\in K(A^\delta)$.
    \item  $g^\pi$ is completely meet-preserving in each positive coordinate, and is completely join-reversing in each negative coordinate. Moreover, $g^\pi[\overline{o}/!\overline x, \overline{k}/!\overline y ]\in O(A^\delta)$.
\end{enumerate}
\end{lemma}
In what follows, we will also omit the superscripts, and rely on the argument for disambiguation. 
%

\begin{lemma}[cf.~\cite{de2024obligations2}, Lemma 4.7]
\label{lem: diamond-output equivalence extended}
For any  spd-algebra $\mathbb{H} = (A, \mathsf{S}, \mathsf{C}, \mathsf{D})$, for all $a, b \in A$, and all $\overline k$ and $ \overline o$,  
\begin{enumerate}
    \item for any $f \in \mathcal{F}$, which is positive in each $x$-variable and negative in each $y$-variable,
    \begin{enumerate}[label=(\roman*)]
        \item $f[\overline k/!\overline x, \overline o/!\overline y]\leq b$ implies $f[\overline a/!\overline x, \overline b/!\overline y]\leq b$ for some $\overline a\geq \overline k$ and some $\overline b\leq \overline o$;
        \item $f[\overline k/!\overline x, \overline o/!\overline y]\leq o$ implies $f[\overline a/!\overline x, \overline b/!\overline y]\leq b$ for some $\overline a\geq \overline k$ some $\overline b\leq \overline o$, and some $b \leq o$;
    \end{enumerate}
    \item for any $g \in \mathcal{G}$, which is positive in each $x$-variable and negative in each $y$-variable,
    \begin{enumerate}[label=(\roman*)]
        \item $a \leq g[\overline o/!\overline x, \overline k/!\overline y]$ implies $a \leq g[\overline b/!\overline x, \overline a/!\overline y]$  for some $\overline a\geq \overline k$ and some $\overline b\leq \overline o$;
        \item $k \leq g[\overline o/!\overline x, \overline k/!\overline y]$ implies $a \leq g[\overline b/!\overline x, \overline a/!\overline y]$  for some $\overline a\geq \overline k$ some $\overline b\leq \overline o$, and some $a \geq k$.
    \end{enumerate}
\end{enumerate}
\end{lemma}
\subsection{The  LE-logic of spd-algebras}
\label{ssec:LE-logic}
For any spd-type $(n_{\mathsf{S}}, n_{\mathsf{C}},n_{\mathsf{D}})$,
the language $\mathcal{L}_\mathrm{LE}(\mathcal{F}, \mathcal{G})$ (from now on abbreviated as $\mathcal{L}_\mathrm{LE}$ or $\mathcal{L}$) takes as parameters: a denumerable set of proposition letters $\mathsf{Prop}$, elements of which are denoted $p,q,r$, possibly with indexes, and disjoint sets of connectives $\mathcal{F}: = \mathcal{F}_A\cup \mathcal{F}_{\mathsf{S}}\cup \mathcal{F}_{\mathsf{D}}$ and $\mathcal{G}: = \mathcal{G}_A\cup\mathcal{G}_{\mathsf{S}}\cup \mathcal{G}_{\mathsf{C}}$ s.t.~$|\mathcal{F}_{\mathsf{S}}| = |\mathcal{G}_{\mathsf{S}}| = n_{\mathsf{S}}$, while $|\mathcal{F}_{\mathsf{D}}| = 2n_{\mathsf{D}}$ and $|\mathcal{G}_{\mathsf{C}}| = 2n_{\mathsf{C}}$. Each $f\in \mathcal{F}_{A}$ and $g\in \mathcal{G}_{A}$ is as indicated in Footnote \ref{ftn:LE-algebra}. Each $f\in \mathcal{F}_{\mathsf{S}}\cup \mathcal{F}_{\mathsf{D}}$ and $g\in \mathcal{G}_{\mathsf{S}}\cup \mathcal{G}_{\mathsf{C}}$ is unary and is associated with some order-type $\varepsilon_f\in \{1, \partial\}$  (resp.~$\varepsilon_g\in \{1, \partial\}$).
Specifically, any $f \in \mathcal{F}_{\mathsf{S}}$ (resp.~$g \in \mathcal{G}_{\mathsf{S}}$) has order-type 1 and is  denoted  $\Diamond$ (resp.~$\blacksquare$). The order type of all the remaining connectives is $\partial$. Connectives in $\mathcal{G}_{\mathsf{C}}$ (resp.~$\mathcal{F}_{\mathsf{D}}$) come in pairs $\wt, \bt$ (resp.~$\tw, \tb$).  The terms (formulas) of $\mathcal{L}_\mathrm{LE}$ are defined recursively as follows:
\[
\varphi ::= p \mid \bot \mid \top \mid \varphi \wedge \varphi \mid \varphi \vee \varphi \mid f(\overline \varphi) \mid g(\overline \varphi)
\]
where $p \in \mathsf{Prop}$, $f\in\mathcal{F}$ and $g\in\mathcal{G}$. Terms in $\mathcal{L}_\mathrm{LE}$ are denoted  by lowercase Greek letters $\varphi, \psi, \gamma$. Terms in the $\mathcal{L}(\mathcal{F}_A, \mathcal{G}_A)$-fragment of $\mathcal{L}_{\mathrm{LE}}$ are  referred to as {\em clopen} terms, and are denoted  by  lowercase letters $s,t$.

\medskip
The {\em basic}, or {\em minimal tense} $\mathcal{L}_\mathrm{LE}$-{\em logic} is a set $\mathbf{L}_\mathrm{LE}$ of $\mathcal{L}_\mathrm{LE}$-sequents $\varphi\vdash\psi$,  which contains  the following sequents:\footnote{\label{ftn:distributive}If the propositional base is {\em distributive}, the minimal {\em DLE-logic} also contains the sequent $p\wedge (q\vee r)\vdash (p\wedge q)\vee (p\wedge r)$. For the sake of a more compact presentation, we let $x \vee^{\epsilon_i}y\coloneqq x\vee y$ if $\epsilon_i = 1$ and $x\vee^{\epsilon_i}y\coloneqq x\wedge y$ if $\epsilon_i = \partial$, and similarly for $\wedge^{\epsilon_i}$, $\top^{\epsilon_i}$ and $\bot^{\epsilon_i}$; we also write $\varphi\leq^{\epsilon_i}\psi$ for $\varphi\leq\psi$ if $\epsilon_i = 1$ and for  $\psi\leq\varphi$ if $\epsilon_i = \partial$, and similarly for $\varphi\vdash^{\epsilon_i}\psi$. Moreover, we write e.g.~$f[x\vee^{\epsilon_f(i)}y]_i$ to indicate that the value of  $f$ in its $i$th coordinate has been instantiated to $x\vee^{\epsilon_f(i)}y$, while the values of the other coordinates have remained unchanged.}
	\begin{align*}
	&p\vdash p, && \bot\vdash p, && p\vdash \top,&&\\
	&p\vdash p\vee q  && q\vdash p\vee q, && p\wedge q\vdash p, && p\wedge q\vdash q,
	\end{align*}
	\begin{align*}
	f[q\vee^{\epsilon_f(i)} r]_i \vdash f[q]_i\vee f[r]_i && 	
	g[q]_i\wedge g[r]_i\vdash g[q\wedge^{\epsilon_g(i)} r]_i &&
	f[\bot^{\epsilon_f(i)}]_i \vdash \bot &&
	 \top\vdash g[\top^{\epsilon_g(i)}]_i\\	
 \end{align*}
	and is closed under the following inference rules:\footnote{\label{ftn:display or residuation}The rules in the second row are referred to as {\em display} or {\em residuation rules}, and can be applied both top-to-bottom and bottom-to-top. The last two rules hold only for those connectives $f\in \mathcal{F}_A$ (resp.~$g\in \mathcal{G}_A$) whose residuals in their $i$th coordinate, denoted $f_i^\sharp$ (resp.~$g_i^\flat$), are included in the given LE-signature.}
	\begin{displaymath}
	\frac{\varphi\vdash \chi\quad \chi\vdash \psi}{\varphi\vdash \psi}
	\qquad
	\frac{\varphi\vdash \psi}{\varphi(\chi/p)\vdash\psi(\chi/p)}
	\qquad
	\frac{\chi\vdash\varphi\quad \chi\vdash\psi}{\chi\vdash \varphi\wedge\psi}
	\qquad
	\frac{\varphi\vdash\chi\quad \psi\vdash\chi}{\varphi\vee\psi\vdash\chi}
	\qquad
	\frac{\varphi\vdash^{\epsilon_{h}(i)}\psi}{h[\varphi]_i\vdash h[\psi]_i}
	\end{displaymath}

\begin{center}
\begin{tabular}{ccccccccc}
 \AxiomC{$\Diamond \varphi \fCenter \psi$}
\doubleLine
\UnaryInfC{$\varphi\fCenter\blacksquare \psi$}
\DisplayProof
& &
 \AxiomC{$\varphi\fCenter \wt\psi$}
\doubleLine
\UnaryInfC{$\psi\fCenter\bt\varphi$}
\DisplayProof
& &
\AxiomC{$\tw\varphi\fCenter \psi$}
\doubleLine
\UnaryInfC{$\tb\psi\fCenter\varphi$}
\DisplayProof

& &
\AxiomC{$f[\varphi]_i\fCenter \psi$}
\doubleLine
\UnaryInfC{$\varphi \vdash^{\epsilon_f(i)}f_i^\sharp[\psi]_i$}
\DisplayProof

& &
\AxiomC{$\varphi\fCenter g[\psi]_i$}
\doubleLine
\UnaryInfC{$g_i^\flat[\varphi]_i \vdash^{\epsilon_g(i)}\psi$}
\DisplayProof

\end{tabular}
\end{center}
	where $h\in\mathcal{F}\cup\mathcal{G}$, and $\Diamond$ and $\blacksquare$ (resp.~$\wt, \bt$, and $\tw, \tb$) have the same index in $n_\mathsf{S}$ (resp.~$n_\mathsf{C}$, $n_\mathsf{D}$). 
	
	\medskip
	
	 We typically drop reference to the parameters when they are clear from the context. By an {\em $\mathrm{LE}$-logic} in the context of spd-algebras we understand any axiomatic extension of $\mathbf{L}_\mathrm{LE}$ in the language $\mathcal{L}_{\mathrm{LE}}$. If all the axioms in the extension are analytic  (cf.~Section \ref{ssec:analytic-LE-axioms}) we say that the given $\mathrm{LE}$-logic is {\em analytic}.

   For any spd-algebra $\mathbb{H}$, an $\mathcal{L}$-inequality (or sequent) $\phi\leq\psi$ is {\em satisfied} in the slanted $\mathcal{L}$-algebra $\mathbb{H}^\ast$ under the assignment $v:\mathsf{Prop}\to A$ (notation: $(\mathbb{H}^\ast, v)\models \phi\leq\psi$) if $((\mathbb{H}^\ast)^\delta, e\cdot v)\models \phi\leq\psi$ in the usual sense, where $e\cdot v$ is the assignment on $A^\delta$ obtained by composing the assignment $v:\mathsf{Prop}\to A$ and the canonical embedding $e: A\to A^\delta$.
 Moreover, $\phi\leq\psi$ is {\em valid} in $\mathbb{H}^\ast$ (notation: $\mathbb{H}^\ast\models \phi\leq\psi$) if $((\mathbb{H}^\ast)^\delta, e\cdot v)\models \phi\leq\psi$ for every assignment $v: \mathsf{Prop}\to A$   (notation: $(\mathbb{H}^\ast)^\delta\models_{A} \phi\leq\psi$).
\subsection{Analytic $\mathcal{L}_{\mathrm{LE}}$-axioms}
\label{ssec:analytic-LE-axioms}
The \emph{positive} (resp.~\emph{negative}) {\em generation tree} of any $\mathcal{L}_\mathrm{LE}$-term $\phi$ is defined by labelling the root node of the generation tree (i.e.~syntax tree) of $\phi$ with the sign $+$ (resp.~$-$), and then propagating the labelling on each remaining node as follows:
	\begin{itemize}
		\item For any node labelled with $ \lor$ or $\land$, assign the same sign to its children nodes.
		\item For any node labelled with $h\in \mathcal{F}\cup \mathcal{G}$ of arity $n_h\geq 1$, and for any $1\leq i\leq n_h$, assign the same (resp.~the opposite) sign to its $i$th child node if $\varepsilon_h(i) = 1$ (resp.~if $\varepsilon_h(i) = \partial$).
	\end{itemize}
	Nodes in signed generation trees are \emph{positive} (resp.~\emph{negative}) if they are signed $+$ (resp.~$-$). Signed generation trees will  mostly be used in the context of term inequalities $\varphi\leq \psi$. In this context, we will typically consider the positive generation tree $+\varphi$ for the left-hand side and the negative one $-\psi$ for the right-hand side. When we consider  $-\varphi$ and $+\psi$, we will speak of the dual generation tree of the term inequality $\varphi\leq \psi$.
 In what follows, we will often need to use {\em placeholder variables} to e.g.~specify the occurrence of a subformula within a given formula. In these cases, we will write e.g.~$\varphi(!z)$ (resp.~$\varphi(!\overline{z})$) to indicate that the variable $z$ (resp.~each variable $z$ in  vector $\overline{z}$) occurs exactly once in $\varphi$. Accordingly, we will write $\varphi[\gamma / !z]$  (resp.~$\varphi[\overline{\gamma}/!\overline{z}]$)   to indicate the formula obtained from $\varphi$ by substituting $\gamma$ (resp.~each formula $\gamma$ in $\overline{\gamma}$) for the unique occurrence of (its corresponding variable) $z$ in $\varphi$. Also, in what follows, we will find it sometimes useful to group placeholder variables together according to certain assumptions we make about them. So, for instance, we will sometimes write e.g.~$\varphi(!\overline{x}, !\overline{y})$  or $f(!\overline{x}, !\overline{y})$ to indicate that $\varphi$ or $f$ is monotone (resp.~antitone) in the coordinates corresponding to every variable $x$ in  $\overline{x}$ (resp.~$y$ in  $\overline{y}$).  We will provide further explanations as to the intended meaning of these groupings whenever required. Finally, we will also extend these conventions to inequalities or sequents, and thus write e.g.~$(\varphi\leq \psi) [\overline{\gamma}/!\overline{z}, \overline{\delta}/!\overline{w}] $ to indicate  the inequality obtained from $\varphi\leq \psi$ by substituting each formula $\gamma$ in $\overline{\gamma}$ (resp.~$\delta$ in $\overline{\delta}$) for the unique occurrence of its corresponding variable $z$ (resp.~$w$) in $\varphi\leq \psi$. Non-leaf nodes in signed generation trees are called \emph{$\Delta$-adjoints}, \emph{syntactically left residuals (SLR)}, \emph{syntactically right residuals (SRR)}, and \emph{syntactically right adjoints (SRA)}, according to the specification given in Table \ref{Join:and:Meet:Friendly:Table} (the occurrences of $\vee$ and $\wedge$ in round brackets signify that these connectives can be listed among the SLR and SRR nodes, as indicated, only if the basic setting is distributive, cf.~Footnote \ref{ftn:distributive}).
	Nodes that are either classified as  $\Delta$-adjoints or SLR are collectively referred to as {\em Skeleton-nodes}, while SRA- and SRR-nodes are referred to as {\em PIA-nodes}. An $\mathcal{L}$-formula $\varphi$ (resp.~$\psi$) is a {\em positive} (resp.~{\em negative}) {\em Skeleton formula} if all nodes in $+\varphi$ and $-\psi$ are Skeleton. A Skeleton formula is {\em definite} if all its nodes are SLR. Likewise, an $\mathcal{L}$-formula $\gamma$ (resp.~$\delta$) is a {\em positive} (resp.~{\em negative}) {\em PIA formula} if all nodes in $+\gamma$  (resp.~$-\delta$) are PIA.
	A branch in a signed generation tree $\ast \phi$, with $\ast \in \{+, - \}$, is called a \emph{good branch} if it is the concatenation of two paths $P_1$ and $P_2$, one of which may possibly be of length $0$, such that $P_1$ is a path from the leaf consisting (apart from variable nodes) only of PIA-nodes, and $P_2$ consists (apart from variable nodes) only of Skeleton-nodes.
\begin{table}[h]
		\begin{center}
			\bgroup
			\def\arraystretch{1.2}
			\begin{tabular}{| c | c |}
				\hline
				Skeleton &PIA\\
				\hline
				$\Delta$-adjoints & Syntactically Right Adjoint (SRA) \\
				\begin{tabular}{ c c c c c c}
					$+$ &$\vee$ &\\
					$-$ &$\wedge$ \\
				\end{tabular}
				&
				\begin{tabular}{c c c c }
					$+$ &$\wedge$ &$\blacksquare\ \ \wt\ \ \bt$ & \\
					$-$ &$\vee$ &$\wdia\ \ \tw\ \ \tb$ &  \\

				\end{tabular}
				\\ \hline
				Syntactically Left Residual (SLR) & Syntactically Right Residual (SRR) \\
				\begin{tabular}{c c c c }
					$+$ &  &$ (\wedge)\ \ \wdia\ \ \tw\ \ \tb\ \ f\in \mathcal{F}_A$ & \\
					$-$ &  &$ (\vee)\ \ \blacksquare\ \ \wt\ \ \bt\ \ g\in \mathcal{G}_A$ & \\
				\end{tabular}
				&\begin{tabular}{c c c c }
				$+$ &  &$ (\vee)\ \ \blacksquare\ \ \wt\ \ \bt\ \ g\in \mathcal{G}_A$ & \\	$-$ &  &$ (\wedge)\ \ \wdia\ \ \tw\ \ \tb\ \ f\in \mathcal{F}_A$ & \\
					
				\end{tabular}
				\\
				\hline
			\end{tabular}
			\egroup
		\end{center}
		\caption{Skeleton and PIA nodes for the  LE-language of spd-algebras.}\label{Join:and:Meet:Friendly:Table}
		\vspace{-1em}
	\end{table}

 An inequality $\varphi \leq \psi$ is \emph{analytic} if every branch of $+\varphi$ and $-\psi$ is good. This implies that  analytic inequalities can  equivalently be represented as $\mathcal{L}$-inequalities $(\varphi\leq \psi)[\overline{\gamma}/!\overline{x}, \overline{\delta}/!\overline{y}]$ such that  $\varphi$ (resp.~$\psi$) is positive (resp.~negative)  Skeleton,  $\gamma$  is positive  PIA for every $\gamma$ in $\overline{\gamma}$, and $\delta$ is negative PIA for  every $\delta$ in $\overline{\delta}$.
 Such an analytic inequality  is \emph{definite} if $+\varphi$ and $-\psi$ only  contain SLR nodes. Special  analytic inequalities are those s.t.~every node in their signed generation trees is Skeleton (resp.~PIA); we refer to these as {\em Skeleton} (resp.~{\em PIA}) {\em inequalities}.
 \begin{example}
 \label{ex:running example}
     The $\mathcal{L}$-inequality \[\tw (\blacksquare \wt\tw \tw \wt v_1\wedge \wt \Diamond \wt v_1)\leq \blacksquare \wt (\bt v_1\vee \blacksquare \bt (\tb \wt v_2\vee \Diamond v_2 ))\] is analytic, since it can be represented as $\varphi[\delta_1/!y_1, \gamma_1/!x_1]\leq \psi[\gamma_2/!x_2, \gamma_3/!x_3]$ with
     \begin{center}
     \begin{tabular}{llll}
     $\varphi(!y_1, !x_1)\coloneqq \tw (\blacksquare \wt\tw  y_1\wedge \wt \Diamond x_1)$ & $\delta_1\coloneqq \tw\wt v_1$ & $\gamma_1\coloneqq \wt v_1$ & \\
     $\psi(!x_2, !x_3)\coloneqq \blacksquare \wt (x_2\vee x_3)$ & $\gamma_2\coloneqq \bt v_1$ & $\gamma_3\coloneqq \blacksquare \bt (\tb \wt v_2\vee \Diamond v_2 )$
     \end{tabular}
     \end{center}
     and it can be easily verified, using the table above, that $\varphi$ is positive Skeleton (i.e.~all nodes in $+\varphi$ are Skeleton), $\psi$ is negative Skeleton (i.e.~all nodes in $-\psi$ are Skeleton), $\delta_1$ is negative PIA and all the $\gamma$s are positive PIA.
 \end{example}
\begin{definition}
\label{def:clopen_analytic}
A {\em clopen-analytic} (resp.~{\em clopen-Skeleton}, {\em clopen-PIA}) {\em inequality}  is an $\mathcal{L}_\mathrm{LE}$-inequality $(\phi \leq \psi)[\overline t/!\overline x]$, s.t.~$(\phi \leq \psi)(!\overline x)$ is an analytic (resp.~Skeleton, PIA) $\mathcal{L}_\mathrm{LE}$-inequality, and every $t$ in $\overline t$ is a clopen term (cf.~Section \ref{ssec:LE-logic}).
\end{definition}
Clopen-analytic inequalities  properly extend analytic inequalities, as shown by the following

\begin{example}
    The DLE-inequality
\[\tw\tw( (v_1\pdla v_2) \vee  \wt (v_1\rightarrow v_2))\leq \blacksquare ((v_1\rightarrow v_2)\wedge \blacksquare (v_1\pdla v_2))\]
is not analytic, since all branches on the left-hand side of the inequality are not good; however, it is  clopen-analytic,  since it can be represented as $\varphi[\delta_1[t_1/!y_4, t_2/!x_2]/!y_1]\leq \psi[\gamma_1/!x_1, \delta_2/!y_2, \delta_3/!y_3]$ with
     \begin{center}
     \begin{tabular}{llll}
     $\varphi(!y_1)\coloneqq \tw y_1$ & $\delta_1 (!y_4, !x_2)\coloneqq \tw (y_4 \vee  \wt x_2)$ & $t_1\coloneqq v_1\pdla v_2$ & \\
     $\psi(!x_1, !y_2, !y_3)\coloneqq \blacksquare ((x_1\rightarrow y_2)\wedge \blacksquare y_3)$ & $\gamma_1\coloneqq  v_1$ & $t_2\coloneqq v_1\rightarrow v_2$ \\
& $\delta_2\coloneqq v_2$ \\
& $\delta_3\coloneqq v_1\pdla v_2$ \\
     \end{tabular}
     \end{center}
For ease of reading, in the representation above we have isolated only the clopen terms which violate the analyticity condition, and have left the non-offending ones as integral parts of $\gamma_1$, $\delta_2$ and $\delta_3$.
\end{example}

In what follows, we will sometimes refer to formulas $\varphi[\overline{t}/!\overline{x}]$ s.t.~$\varphi(!\overline{x})$ is a positive (resp.~negative) Skeleton formula and every $t$ in $\overline{t}$ is a clopen term as  {\em positive} (resp.~{\em negative}) {\em clopen-Skeleton formulas}.


Recall that the unique homomorphic extension of  every valuation $v: \mathsf{Prop}\to A$ is a map $v: \mathcal{L}\to A^\delta$.

\begin{proposition}(cf.~\cite{conradie2020constructive} Lemma 5.2)
\label{prop: definite skeleton implies preservation properties}
    For any spd-algebra $\mathbb{H} = (A, \mathsf{S}, \mathsf{C}, \mathsf{D})$, all $\mathcal{L}$-formulas $\varphi(!\overline{x}, !\overline{y})$, $ \psi(!\overline{y}, !\overline{x})$ s.t.~$\varphi$ is definite positive Skeleton, $\psi$ is definite negative Skeleton, and the signed generation tree of $(\varphi\leq \psi)[!\overline{x}, !\overline{y}]$ is positive in each $x$ in $\overline{x}$ and negative in each $y$ in $\overline{y}$,
\begin{enumerate}
\item the term function corresponding to $\varphi$ on $\mathbb{H}^\delta$   is completely join-preserving, hence monotone, in each $x$ in $\overline{x}$ and completely meet-reversing, hence antitone, in each $y$ in $\overline{y}$.

    \item the term function corresponding to $\psi$ on $\mathbb{H}^\delta$   is completely meet-preserving, hence monotone, in each $y$ in $\overline{y}$ and completely join-reversing, hence antitone, in each $x$ in $\overline{x}$.

    \end{enumerate}
\end{proposition}
\begin{proof}
    By simultaneous induction  on the complexity of $\varphi$ and $  \psi$. If $\varphi\coloneqq x$, $\psi\coloneqq y$, then the statements are trivially true. Let $\varphi\coloneqq f(\overline{\varphi'}, \overline{\psi'})$ for $f\in \mathcal{F} $, where each $\varphi'$  in $\overline{\varphi'}$ (resp.~$\psi'$ in $\overline{\psi'}$) is a definite positive (resp.~negative) Skeleton formula; hence, by inductive hypothesis, each $\varphi'$ and each $\psi'$ satisfies the statement of the proposition. Moreover, by Lemma \ref{lemma: distribution properties of sigma pi}, $f$ is completely join-preserving in each positive coordinate and completely meet-reversing in each negative coordinate. Hence, the following chain of identities hold:
    \begin{center}
        \begin{tabular}{r cl l}
          $\varphi [\overline{\bigvee a_i}/!\overline{x}, \overline{\bigwedge b_i}/!\overline{y}]$   & $=$ & $f(\overline{\varphi'}, \overline{\psi'})[\overline{\bigvee a_i}/!\overline{x}, \overline{\bigwedge b_i}/!\overline{y} ] $\\
           & $=$ & $f(\overline{\varphi'[\overline{\bigvee a_i}/!\overline{x}, \overline{\bigwedge b_j}/!\overline{y}]}, \overline{\psi'[ \overline{\bigwedge b_j}/!\overline{y},\overline{\bigvee a_i}/!\overline{x}]})$ \\
             & $=$ & $f(\overline{\bigvee\varphi'(a_i, b_j)}, \overline{\bigwedge \psi'( b_j, a_i)})$ & induction hypothesis \\
             & $=$ & $\bigvee f(\overline{\varphi'(a_i, b_j)}, \overline{ \psi'( b_j, a_i)})$ & Lemma \ref{lemma: distribution properties of sigma pi} \\
             & $=$ & $\bigvee \varphi[\overline{a_i}/!\overline{x}, \overline{b_j}/!\overline{y}]$. & \\
        \end{tabular}
    \end{center}
The case in which $\psi\coloneqq g( \overline{\psi'}, \overline{\varphi'})$ for $g\in \mathcal{G}$ is analogous.
\end{proof}

The next proposition accounts for why terms in the $\mathcal{L}(\mathcal{F}_A, \mathcal{G}_A)$-fragment are referred to as clopen terms.
\begin{proposition}
\label{prop: positive PIA are open}
For any spd-algebra $\mathbb{H} = (A, \mathsf{S}, \mathsf{C}, \mathsf{D})$, all $\mathcal{L}$-formulas $\varphi, \psi$, any clopen $\mathcal{L}$-term $t$  and any valuation $v: \mathsf{Prop}\to A$,
\begin{enumerate}
  \item  $v(t) \in A$.
    \item If $\psi$ is positive (clopen) PIA, i.e.~negative (clopen) Skeleton, then $v(\psi)\in O(A^\delta)$.
    \item If $\varphi$ is positive (clopen) Skeleton, i.e.~negative (clopen) PIA, then $v(\varphi)\in K(A^\delta)$.

\end{enumerate}
\end{proposition}
\begin{proof}
Item (1) is proven by straightforward induction on the shape of clopen terms.
To prove (2) and (3), we proceed by simultaneous induction on the complexity of $\varphi$ and $\psi$.
The base case holds by item (1).

If $\psi \coloneqq \psi_1 \wedge \psi_2$, then $\psi_1 $ and $\psi_2$ are  positive (clopen) PIA formulas of lower complexity than $\psi$, hence by inductive hypothesis $v(\psi_1), v(\psi_2)\in O(A^\delta)$. Then, by Proposition \ref{prop:background can-ext}(iv), $v(\psi) = v(\psi_1 \wedge \psi_2) = v(\psi_1)\wedge v(\psi_2)\in O(A^\delta)$. The case of $\varphi \coloneqq \varphi_1 \vee \varphi_2$ is shown similarly, using the remaining part of Proposition \ref{prop:background can-ext}(iv).

If $\psi \coloneqq \blacksquare \psi'$, then $\psi'$ is a  positive (clopen) PIA formula of lower complexity than $\psi$, hence by inductive hypothesis  $v(\psi')\in O(A^\delta)$. Then, by Lemma \ref{lemma: distribution properties of sigma pi}.2, $v(\psi) = v(\blacksquare \psi') = \blacksquare v(\psi') \in O(A^\delta)$.  The case of $\varphi \coloneqq \wdia \varphi'$ is shown similarly, using Lemma \ref{lemma: distribution properties of sigma pi}.1.

If $\psi \coloneqq \wt \varphi$, then $\varphi$ is a negative (clopen) PIA formula of lower complexity than $\psi$, hence by inductive hypothesis,  $v(\varphi)\in K(A^\delta)$. Then, by Lemma \ref{lemma: distribution properties of sigma pi}.2, $v(\psi) = v(\wt \varphi) = \wt v(\varphi) \in O(A^\delta)$.
The remaining cases  are treated analogously, using Lemma \ref{lemma: distribution properties of sigma pi}.
\end{proof}

\section{Correspondence}
\label{sec:correspondence}

In the present section, we fix an arbitrary spd-type and its associated language $\mathcal{L} = \mathcal{L}_{\mathrm{LE}}(\mathcal{F}, \mathcal{G})$, and show that any clopen-analytic $\mathcal{L}$-inequality  (cf.~Definition \ref{def:clopen_analytic}) {\em corresponds} to some $\mathcal{L}^{\mathrm{FO}}_{spd}$-formulas in the sense specified in the following theorem. Moreover, the $\mathcal{L}^{\mathrm{FO}}_{spd}$-formulas  can effectively be computed from the input $\mathcal{L}$-inequality.  

\begin{theorem}
  For every spd-algebra $\mathbb{H}$ and every clopen-analytic $\mathcal{L}$-inequality $(\varphi\leq \psi)[\overline{\gamma}/!\overline{x}, \overline{\delta}/!\overline{y}]$,
\[\mathbb{H}^\ast\models (\varphi\leq \psi)[\overline{\gamma}/!\overline{x}, \overline{\delta}/!\overline{y}] \quad\text{ iff }\quad \mathbb{H}\models \mathsf{FLAT}(\mathsf{ALBA}((\varphi\leq \psi)[\overline{\gamma}/!\overline{x}, \overline{\delta}/!\overline{y}]))^\tau,\]
where $\mathsf{ALBA}((\varphi\leq \psi)[\overline{\gamma}/!\overline{x}, \overline{\delta}/!\overline{y}])$ is the output of the algorithm  $\mathsf{ALBA}$ (cf.~Algorithm \ref{algo:alba}) described  in Steps 1-4 below, $\mathsf{FLAT}$ (cf.~Algorithm \ref{algo:flat}) is described in Step 5, and $\tau$ is the translation from flat $\mathcal{L}$-inequalities  to $\mathcal{L}^{\mathrm{FO}}_{spd}$-formulas described in Step 6.
\end{theorem}

In what follows, we prove the theorem above for analytic inequalities, and we then extend it to clopen-analytic inequalities in Remark \ref{remark:extend_to_clopen}.  Also, we will use $\overline{x}$ (resp.~$\overline{y}$) as placeholder variables which are positive (resp.~negative) in the signed generation tree of the input $\mathcal{L}$-inequalities, and we assume that $\overline v$ is the vector of variables occurring in the input analytic inequality $(\varphi\leq \psi)[\overline{\gamma}/!\overline{x}, \overline{\delta}/!\overline{y}]$.

\paragraph{Step 1 (preprocessing).}

$\mathsf{ALBA}$ receives an analytic inequality $(\varphi\leq \psi)[\overline{\gamma}/!\overline{x}, \overline{\delta}/!\overline{y}]$ as input and  exhaustively distributes  $ f \in \mathcal{F}$  over $\vee$, and $g \in \mathcal{G}$ over $\wedge$. Hence, occurrences of $+\vee$ and $-\wedge$ which are brought to the roots of the generation trees can be eliminated  via exhaustive applications of \textit{splitting} rules.

\textbf{Splitting-rules}.
\begin{center}
    \begin{tabular}{ll}
    \infer[]{\varphi \leq \psi_1\quad \varphi \leq \psi_2}{\varphi \leq \psi_1 \wedge \psi_2 }      &
     \infer[]{\varphi_1\leq \psi \quad \varphi_2 \leq \psi}{\varphi_1 \vee \varphi_2 \leq \psi }
    \end{tabular}
\end{center}

This gives rise to a finite set of {\em definite} analytic inequalities (cf.~discussion after Table \ref{Join:and:Meet:Friendly:Table}), each of which of the form $(\varphi \leq \psi)[\overline{\alpha}/!\overline{x}, \overline{\beta}/!\overline{y}]$, where each $\alpha$ in $\overline{\alpha}$ (resp.~$\beta$  in $\overline{\beta}$) is  of the form $\bigwedge_i \gamma_i$ (resp.~$\bigvee_i \delta_i$), and each $\gamma_i$ (resp.~$\delta_i$) is a definite positive (resp.~negative) PIA formula (cf.~discussion before Table \ref{Join:and:Meet:Friendly:Table}). Each element of this set is passed separately to step 2.

\paragraph{Step 2 (first approximation).}

Let $(\varphi \leq \psi)[\overline{\alpha}/!\overline{x}, \overline{\beta}/!\overline{y}]$ be a definite analytic inequality as above.
In what follows, variables $k$  range in $ K(A^\delta)$, $o$  in $O(A^\delta)$, and  $a, b$  in $A$.  For the sake of improved readability, in what follows we use the same notations both for formulas and for their interpretations in spd-algebras under valuations $v$ into $A$.
\begin{lemma}[Soundness of first approximation]
\label{lemma:first_approx}
    For any spd-algebra $\mathbb{H}$ and every definite analytic inequality $(\varphi \leq \psi)[\overline{\alpha}/!\overline{x}, \overline{\beta}/!\overline{y}]$,

\begin{center}
    \begin{tabular}{rll}
         & $\mathbb{H}^\ast\models (\varphi \leq \psi)[\overline{\alpha}/!\overline{x}, \overline{\beta}/!\overline{y}]$ & \\
    iff & $\mathbb{H}^\ast\models \forall \overline{k}\forall \overline{o}( \overline{k\leq \alpha}\ \&\ \overline{\beta\leq o}\Rightarrow (\varphi \leq \psi)[\overline{k}/!\overline{x}, \overline{o}/!\overline{y}]) $ &\\
    \end{tabular}
\end{center}
\end{lemma}
\begin{proof}
    By denseness, $\alpha = \bigvee\{ k \mid k\leq\alpha\}$ (resp.~$\beta = \bigwedge\{ o \mid \beta \leq o\}$) for any $\alpha$ in $\overline{\alpha}$ (resp.~$\beta$  in $\overline{\beta}$); hence,
    \begin{center}
    \begin{tabular}{rll}
         & $\mathbb{H}^\ast\models (\varphi \leq \psi)[\overline{\alpha}/!\overline{x}, \overline{\beta}/!\overline{y}]$ & \\
    iff & $\mathbb{H}^\ast\models (\varphi \leq \psi)[\overline{\bigvee\{ k \mid k\leq\alpha\}}/!\overline{x}] ,\overline{\bigwedge\{ o \mid \beta \leq o\}}/!\overline{y}]$ &\\
    iff & $\mathbb{H}^\ast\models \bigvee\{ \varphi[\overline{k}/!\overline{x}, \overline{o}/!\overline{y}] \mid \overline{k\leq\alpha}, \overline{\beta \leq o} \} \leq \bigwedge\{ \psi [\overline{o}/!\overline{y},\overline{k}/!\overline{x}]  \mid \overline{k\leq\alpha}, \overline{\beta \leq o} \}$ &\\
    \end{tabular}
\end{center}
    The last equivalence holds since, by Proposition \ref{prop: definite skeleton implies preservation properties}, $\varphi$ (resp.~$\psi$)  is completely join-preserving (resp.~join-reversing) in its $\overline{x}$-coordinates  and completely meet-reversing (resp.~meet-preserving) in its $\overline{y}$-coordinates.
    Hence, by definition of suprema and infima,
    \begin{center}
    \begin{tabular}{rll}
     & $\mathbb{H}^\ast\models \bigvee\{ \varphi[\overline{k}/!\overline{x}, \overline{o}/!\overline{y}] \mid \overline{k\leq\alpha}, \overline{\beta \leq o} \} \leq \bigwedge\{ \psi [\overline{o}/!\overline{y},\overline{k}/!\overline{x}]  \mid \overline{k\leq\alpha}, \overline{\beta \leq o} \}$  &\\
      iff   & $\mathbb{H}^\ast\models \forall \overline{k}\forall \overline{o}( \overline{k\leq \alpha}\ \&\ \overline{\beta\leq o}\Rightarrow (\varphi \leq \psi)[\overline{k}/!\overline{x}, \overline{o}/!\overline{y}]) $.&\\
    \end{tabular}
\end{center}
\end{proof}


\paragraph{Step 3 (compactness).}
\begin{lemma}
    For any spd-algebra $\mathbb{H}$ and every definite analytic $\mathcal{L}$-inequality $(\varphi \leq \psi)[\overline{\alpha}/!\overline{x}, \overline{\beta}/!\overline{y}]$,
\begin{center}
    \begin{tabular}{rll}
     & $\mathbb{H}^\ast\models \mathbb{H}^\ast\models \forall \overline{k}\forall \overline{o}( \overline{k\leq \alpha}\ \&\ \overline{\beta\leq o}\Rightarrow (\varphi \leq \psi)[\overline{k}/!\overline{x}, \overline{o}/!\overline{y}]) $ &\\
iff & $\mathbb{H}^\ast\models \mathbb{H}^\ast\models \forall \overline{k}\forall \overline{o}\forall \overline{a}\forall \overline{b}( \overline{k\leq a\leq \alpha}\ \&\ \overline{\beta\leq b\leq o}\Rightarrow (\varphi \leq \psi)[\overline{k}/!\overline{x}, \overline{o}/!\overline{y}]) $ &\\
 \end{tabular}
\end{center}
\end{lemma}
\begin{proof}
     The assumptions imply that each $\alpha$ in $\overline{\alpha}$ (resp.~$\beta$ in $\overline{\beta}$) is a positive (resp.~negative) PIA. Hence, by Proposition \ref{prop: positive PIA are open},  $\alpha \in O(A^\delta)$ (resp.~$\beta \in K(A^\delta)$). Therefore, by compactness,

\begin{center}
    \begin{tabular}{rll}
     & $\mathbb{H}^\ast\models \mathbb{H}^\ast\models \forall \overline{k}\forall \overline{o}( \overline{k\leq \alpha}\ \&\ \overline{\beta\leq o}\Rightarrow (\varphi \leq \psi)[\overline{k}/!\overline{x}, \overline{o}/!\overline{y}]) $ &\\
     iff & $\mathbb{H}^\ast\models \forall \overline{k}\forall \overline{o}  ( \overline{\exists a(k\leq a\leq \alpha)}\ \&\  \overline{\exists b(\beta \leq b\leq o)}\Rightarrow (\varphi \leq \psi)[\overline{k}/!\overline{x}, \overline{o}/!\overline{y}]) $ &\\
     iff & $\mathbb{H}^\ast\models\forall \overline{k}\forall \overline{o}\forall \overline{a}\forall \overline{b}( \overline{k\leq a\leq \alpha}\ \&\  \overline{\beta \leq b\leq o}\Rightarrow (\varphi \leq \psi)[\overline{k}/!\overline{x}, \overline{o}/!\overline{y}])$. & \\
    \end{tabular}
\end{center}
\end{proof}

\paragraph{Step 4 (Ackermann lemma).}
\begin{lemma}
    For any spd-algebra $\mathbb{H}$ and every definite analytic $\mathcal{L}$-inequality $(\varphi \leq \psi)[\overline{\alpha}/!\overline{x}, \overline{\beta}/!\overline{y}]$,

\begin{center}
    \begin{tabular}{rll}
     &
$\mathbb{H}^\ast\models\forall \overline{k}\forall \overline{o}\forall \overline{a}\forall \overline{b}( \overline{k\leq a\leq \alpha}\ \&\  \overline{\beta \leq b\leq o}\Rightarrow (\varphi \leq \psi)[\overline{k}/!\overline{x}, \overline{o}/!\overline{y}])$ & \\
     iff & $\mathbb{H}^\ast\models\forall \overline{a}\forall \overline{b}( \overline{ a\leq \alpha}\ \&\  \overline{\beta \leq b }\Rightarrow (\varphi \leq \psi)[\overline{a}/!\overline{x}, \overline{b}/!\overline{y}]$ & \\
    \end{tabular}
\end{center}
\end{lemma}
\begin{proof}
The `only if' direction is obtained  by instantiating each $k:=a$ and each $o:=b$, which is allowed, since $A\subseteq K(A^\delta)\cap O(A^\delta)$. To show the `if' direction, let $\overline{a}$, $\overline{b}$, $\overline{k}$, $\overline{o}$ be suitable vectors of elements in $A$, $K(A^\delta)$ and $O(A^\delta)$ s.t.~$\overline{k\leq a \leq \alpha}$ and $\overline{\beta\leq b\leq o}$. The assumptions imply (cf.~Proposition \ref{prop: definite skeleton implies preservation properties}) that the term function $\varphi(!\overline{x}, !\overline{y})$ (resp.~$\psi(!\overline{y}, !\overline{x})$) is monotone (resp.~antitone) in the $\overline{x}$-variables, and antitone (resp.~monotone) in the $\overline{y}$-variables; moreover, $(\varphi \leq \psi)[\overline{a}/!\overline{x}, \overline{b}/!\overline{y}]$. Hence,
    \begin{center}
   $\varphi [\overline{k}/!\overline{x}, \overline{o}/!\overline{y}] \leq \varphi [\overline{a}/!\overline{x}, \overline{b}/!\overline{y}]\leq \psi [\overline{b}/!\overline{y},\overline{a}/!\overline{x}]\leq \psi [ \overline{o}/!\overline{y},\overline{k}/!\overline{x}]$,
\end{center}
which shows that $(\varphi \leq \psi)[\overline{k}/!\overline{x}, \overline{o}/!\overline{y}]$, as  required.
\end{proof}
\paragraph{Step 5 (flattification).}
An  $\mathcal{L}$-inequality is {\em flat} if it is of the form $a\leq b$, or $f(\overline{a})\leq b$, or $a\leq g(\overline{b})$ where $a, b$ and all variables in $\overline{a}$ and $\overline{b}$ range in $\mathsf{Prop}$, $f\in \mathcal{F}$, and $g\in \mathcal{G}$. In what follows, we will generically represent flat inequalities as $\xi(a, \overline{b})$, where the first coordinate is the variable `in display', i.e.~one which occurs in isolation on one side of the inequality.
\begin{definition}
\label{def:big xi}
For any definite positive (resp.~negative) skeleton   $\mathcal{L}$-formula $\varphi(!\overline{x}, !\overline{y})$ (resp.~$\psi(!\overline{y}, !\overline{x})$), the inequality $\varphi[\overline a/!\overline{x}, \overline{b}/!\overline{y}] \leq b$ (resp.~$b\leq \psi[\overline{b}/!\overline{y}, \overline a/!\overline{x}]$) is associated with the existentially quantified conjunction of flat inequalities $\Xi_\varphi(b, \overline a, \overline b)$ (resp.~$\Xi_\psi(b, \overline b, \overline a)$)  defined by simultaneous induction on $\varphi$ and $\psi$ as follows:
\begin{enumerate}
    \item if $\varphi \coloneqq x$ (resp.~$\psi\coloneqq y$), then $\Xi_{\varphi}(a, b)\coloneqq a\leq b$ (resp.~$\Xi_{\psi}(b, a)\coloneqq b\leq a$);
\item if $\varphi \coloneqq f[\overline{\varphi[\overline a/!\overline{x}, \overline{b}/!\overline{y}]}/!\overline{z}, \overline{\psi[ \overline{b}/!\overline{y}, \overline a/!\overline{x}]}/!\overline{z'}]$ (resp.~$\psi\coloneqq g[\overline{\psi[ \overline{b}/!\overline{y}, \overline a/!\overline{x}]}/!\overline{z'}, \overline{\varphi[\overline a/!\overline{x}, \overline{b}/!\overline{y}]}/!\overline{z}]$), then
\[
\Xi_{\varphi}
(b, \overline a, \overline b)
\coloneqq
\exists \overline{u^+} \exists \overline{u^-} \left( f(\overline{u^+}, \overline{u^-})\leq b, \ \& \ \bigwith_i \Xi_{\varphi_i}(u^+_i, \overline a, \overline b) \ \& \ \bigwith_i \Xi_{\psi_i}(u^-_i, \overline b, \overline a)\right)
\]
\[
\Xi_{\psi}
(b, \overline b, \overline a)
\coloneqq
\exists \overline{u^+} \exists \overline{u^-} \left( b\leq g(\overline{u^+}, \overline{u^-}), \ \& \ \bigwith_i \Xi_{\psi_i}(u^+_i, \overline b, \overline a)\ \& \ \bigwith_i \Xi_{\varphi_i}(u^-_i, \overline a, \overline b) \right)
\]
\end{enumerate}
\end{definition}

\begin{lemma}
\label{lemma: flattification of definite skeleton into chains}
For any spd-algebra $\mathbb{H}$ and  any definite positive (resp.~negative) skeleton   $\mathcal{L}$-formula $\varphi(!\overline{x}, !\overline{y})$ (resp.~$\psi(!\overline{y}, !\overline{x})$),
\begin{enumerate}
    \item $\mathbb{H}^\ast\models \varphi[\overline a/!\overline{x}, \overline{b}/!\overline{y}] \leq b\quad $ iff $\quad \mathbb{H}^\ast\models \Xi_\varphi(b, \overline a, \overline b)$
\item $\mathbb{H}^\ast\models b\leq \psi[\overline{b}/!\overline{y}, \overline a/!\overline{x}]\quad $ iff $\quad\mathbb{H}^\ast\models \Xi_\psi(b, \overline b, \overline a)$
\end{enumerate}
\end{lemma}
\begin{proof}
    By simultaneous induction on  $\varphi$ and $\psi$.
    If $\varphi: = x$ or $\psi\coloneqq  y$, the statement is  immediately true. As to the inductive step, we proceed by cases.
    if $\varphi: = \tw \psi$, then    $\varphi[\overline a/!\overline x, \overline b/!\overline y]: = \tw \psi[\overline b/!\overline y,  \overline a/!\overline x]$ for some  negative Skeleton formula $\psi$ of  strictly lower complexity.  Hence, by Proposition \ref{prop: positive PIA are open}.2, $\psi[\overline b/!\overline y,  \overline a/!\overline x]\in O(A^\delta)$. Therefore,
 \begin{center}
     \begin{tabular}{r c l l}
      $\mathbb{H}^\ast\models \tw \psi[\overline b/!\overline y,  \overline a/!\overline x]\leq b$ & iff & $\mathbb{H}^\ast\models \exists u_1 (\tw u_1\leq b\ \&\ u_1\leq \psi[\overline b/!\overline y,  \overline a/!\overline x])$ & Lemma \ref{lem: diamond-output equivalence extended}.1(i) \\
      & iff & $\mathbb{H}^\ast\models \exists u_1 (\tw u_1\leq b\ \&\ \Xi_\psi(u_1, \overline b, \overline a))$ & inductive hypothesis \\
    & iff & $\mathbb{H}^\ast\models \Xi_\varphi(b, \overline a, \overline{b})$, & Definition \ref{def:big xi} \\
     \end{tabular}
 \end{center}
 as required.  The remaining cases are proved similarly, using Proposition \ref{prop: positive PIA are open} and Lemma \ref{lem: diamond-output equivalence extended}.
\end{proof}
   In what follows, we  let $\Xi'(b, \overline u, \overline{b},\overline a)$ denote $\Xi(b, \overline{b},\overline a)$ stripped of the existential quantification.

\begin{lemma}
For any spd-algebra $\mathbb{H}$ and every definite analytic $\mathcal{L}$-inequality $(\varphi \leq \psi)[\overline{\alpha}/!\overline{x}, \overline{\beta}/!\overline{y}]$,

\begin{center}
    \begin{tabular}{rll}

 & $\mathbb{H}^\ast\models\forall \overline v\forall \overline{a}\forall \overline{b}( \overline{ a\leq \alpha}\ \&\  \overline{\beta \leq b }\Rightarrow (\varphi \leq \psi)[\overline{a}/!\overline{x}, \overline{b}/!\overline{y}])$ & \\
iff &$\mathbb{H}^\ast\models\forall \overline v\forall \overline a \forall \overline b \forall \overline c \exists  d \exists \overline d( \overline{\bigwith_i \Xi'_{\gamma_i}(a,\overline c, \overline v, \overline v)}\ \&\ \overline{\bigwith_j \Xi'_{\delta_j}(b, \overline c, \overline v, \overline v)}\Rightarrow (\Xi'_{\varphi}(d, \overline d, \overline a, \overline b) \ \&\ \Xi'_{\psi}(d, \overline d, \overline b, \overline a)))$.&\\

    \end{tabular}
\end{center}

\end{lemma}
\begin{proof}
     By assumption,  $\varphi [\overline{\alpha}/!\overline{x}, \overline{\beta}/!\overline{y}]$ is definite positive Skeleton (hence, by Proposition \ref{prop: positive PIA are open}, $\varphi [\overline{\alpha}/!\overline{x}, \overline{\beta}/!\overline{y}]\in K(A^\delta)$), $\psi[\overline{\beta}/!\overline{y}, \overline{\alpha}/!\overline{x}]$ is definite negative Skeleton (hence, by Proposition \ref{prop: positive PIA are open}, $\psi [\overline{\beta}/!\overline{y}, \overline{\alpha}/!\overline{x}]\in O(A^\delta)$), and $\alpha = \bigwedge_i\gamma_i$ for each $\alpha$ in $\overline \alpha$, with each $\gamma_i$ definite positive PIA (i.e.~definite negative Skeleton) containing only variables in $\overline v$, and $\beta = \bigvee_j\delta_j$ for each $\beta$ in $\overline \beta$ with each $\delta_j$ definite negative PIA (i.e.~definite positive Skeleton) containing only variables in $\overline v$. Thus,
    \begin{center}
    \begin{tabular}{rll}
      & $\mathbb{H}^\ast\models\forall \overline v \forall \overline{a}\forall \overline{b}( \overline{ a\leq \alpha}\ \&\  \overline{\beta \leq b }\Rightarrow (\varphi \leq \psi)[\overline{a}/!\overline{x}, \overline{b}/!\overline{y}])$ & \\
     iff  & $\mathbb{H}^\ast\models\forall \overline v\forall a \forall b ( \overline{\bigwith_i a\leq \gamma_i}\ \&\ \overline{\bigwith_j \delta_j\leq b}\Rightarrow (\varphi \leq \psi)[\overline{a}/!\overline{x}, \overline{b}/!\overline{y}])$ & def.~$\bigwedge$, $\bigvee$ \\
     iff  & $\mathbb{H}^\ast\models\forall \overline v\forall \overline a \forall \overline b ( \overline{\bigwith_i \Xi_{\gamma_i}(a, \overline v, \overline v)}\ \&\ \overline{\bigwith_j \Xi_{\delta_j}(b, \overline v, \overline v)}\Rightarrow \exists d(\varphi [\overline{a}/!\overline{x}, \overline{b}/!\overline{y}] \leq d\ \&\ d\leq \psi[\overline{b}/!\overline{y}, \overline{a}/!\overline{x}]))$ & compactness \\
     iff  & $\mathbb{H}^\ast\models \forall \overline v\forall \overline a \forall \overline b ( \overline{\bigwith_i \Xi_{\gamma_i}(a, \overline v, \overline v)}\ \&\ \overline{\bigwith_j \Xi_{\delta_j}(b, \overline v, \overline v)}\Rightarrow \exists d(\Xi_{\varphi}(d, \overline a, \overline b) \ \&\ \Xi_{\psi}(d, \overline b, \overline a)))$ & Lemma \ref{lemma: flattification of definite skeleton into chains} \\
     iff  & $\mathbb{H}^\ast\models\forall \overline v\forall \overline a \forall \overline b \forall \overline c \exists  d \exists \overline d( \overline{\bigwith_i \Xi'_{\gamma_i}(a,\overline c, \overline v, \overline v)}\ \&\ \overline{\bigwith_j \Xi'_{\delta_j}(b, \overline c, \overline v, \overline v)}\Rightarrow (\Xi'_{\varphi}(d, \overline d, \overline a, \overline b) \ \&\ \Xi'_{\psi}(d, \overline d, \overline b, \overline a)))$ & fresh var's \\
    \end{tabular}
\end{center}
\end{proof}


\paragraph{Step 6 (translation):} Finally, every flat inequality occurring in any $\Xi'$ can be translated to an atomic $\mathcal{L}^{\mathrm{FO}}_{spd}$-formula as indicated in the following table (in which $f\in \mathcal{F}_A$ and $g\in \mathcal{G}_A$):
\smallskip

{{\centering
\begin{tabular}{c c c||ccc}
$\tau(f(\overline a, \overline b)\leq b)$  & $=$ & $f(\overline a, \overline b)\leq b$ &  $\tau(a\leq g(\overline b,\overline a))$ & $=$ & $a\leq g(\overline b,\overline a)$  \\

   $\tau(\Diamond a\leq b)$  & $=$ & $a\prec b$ &  $\tau(a\leq \blacksquare b)$ & $=$ & $a\prec b$  \\
   $\tau(b \leq \wt a)$  & $=$ & $a\npcon b$ &   $\tau(a\leq \bt b)$ &$=$ & $a\npcon b$  \\
     $\tau(\tw b \leq  a)$  & $=$ & $a\nppcon b$   & $\tau(\tb a\leq  b)$ & $=$ & $a\nppcon b$  \\
\end{tabular}
\par}}

 In what follows, we let $\Xi^\tau(b, \overline u, \overline{b}, \overline a)$ denote the result of replacing all flat  inequalities in $\Xi'(b, \overline u, \overline{b}, \overline a)$ with their translations.

\paragraph{Simplifications.}
In case $\varphi(!x)\coloneqq x$ and every $\gamma$ and $\delta$ in $\psi[\overline{\delta}/!\overline{y}, \overline{\gamma}/!\overline{x} ]$ is atomic (resp.~$\psi(!y)\coloneqq y$ and every $\gamma$ and $\delta$ in $\varphi[\overline{\gamma}/!\overline{x}, \overline{\delta}/!\overline{y}]$ is atomic), some steps in the procedure can be  simplified or skipped altogether. Specifically, after step 1, each definite analytic inequality returned at the end of step 1 is of the form $\alpha(\overline v)\leq \psi (\overline v)$ (resp.~$\varphi(\overline v)\leq \beta (\overline v)$), with $\alpha\coloneqq \bigwedge_i \gamma_i(\overline v)$ (resp.~$\beta\coloneqq \bigvee \delta_j(\overline v)$) with each $\gamma_i$ (resp.~$\delta_j$) definite positive (resp.~negative) PIA and $\psi$ (resp.~$\varphi$) definite negative (resp.~positive) Skeleton, and since every variable in $\overline v$ ranges in $A$, by Proposition \ref{prop: positive PIA are open},  $\alpha(\overline v), \psi (\overline v)\in O(A^\delta)$ (resp.~$\beta(\overline v), \varphi(\overline v) \in K(A^\delta)$).
Hence, rather than performing first approximation using closed and open variables and the general denseness, we can use the definition that the set of closed  (resp.~open) elements is the meet-closure (resp.~join-closure) of $A$, and hence perform first approximation by introducing fresh variables ranging in $A$ as follows (cf.~\cite[Proposition 2.6.1(i) and (ii)]{de2024obligations2}):
\begin{center}
\begin{tabular}{cc}
    \begin{tabular}{rll}
         & $\mathbb{H}^\ast\models \alpha\leq \psi $ & \\
    iff  & $\mathbb{H}^\ast\models\forall a  ( \bigwith_i a\leq \gamma_i\Rightarrow a\leq \psi )$ & \\
    \end{tabular}
    &
     \begin{tabular}{rll}
         & $\mathbb{H}^\ast\models \varphi\leq \beta$ & \\
    iff  & $\mathbb{H}^\ast\models\forall b  ( \bigwith_j  \delta_i\leq b\Rightarrow \varphi \leq b)$ & \\
    \end{tabular}
    \end{tabular}
\end{center}
thus  steps 3 and 4 can be skipped. Moreover, in step 5, there is no need to apply compactness to the consequent inequality of the quasi inequality, since it is already in a suitable shape for the application of Lemma \ref{lemma: flattification of definite skeleton into chains}.

\begin{remark}
\label{remark:extend_to_clopen}
The steps above can be applied also to clopen-analytic inequalities  $(\varphi \leq \psi)[\overline t / !\overline x]$ (cf.\ Definition \ref{def:clopen_analytic}). Indeed, by Proposition \ref{prop: positive PIA are open}.1, $v(t) \in A$ for any valuation $v: \mathsf{Prop} \to A$, and any clopen term $t$.
For any such inequality, the correspondent resulting from the procedure above is the conjunction of the formulas in $\mathsf{FLAT}(\mathsf{ALBA}((\varphi \leq \psi)(!\overline x)))^\tau[\overline t / !\overline x]$, i.e., the first order correspondent resulting from applying the procedure above to $(\varphi \leq \psi)(!\overline x)$, where the placeholders in $\overline x$ are then replaced by the terms in $\overline t$.
\end{remark}

\begin{algorithm}
\caption{Algorithm $\mathsf{ALBA}$}
\label{algo:alba}
\small
\hspace*{\algorithmicindent} \textbf{Input}: An analytic inequality $\varphi \leq \psi$ \\
\hspace*{\algorithmicindent} \textbf{Output}: a set of quasi-inequalities encoded as a pair containing the set of inequalities occurring in the antecedent, and the inequality in the consequent.
\begin{algorithmic}[1]
\Function{$\mathsf{ALBA}$}{$\varphi$, $\psi$}
    \State {\textbf{let} $\bigvee^n_i\varphi_i\leq \bigwedge^m_j \psi_j$ the inequality obtained by distributing $\vee$ and $\wedge$ in $\varphi\leq\psi$ exhaustively}
    \State {\textbf{let} $output \leftarrow \varnothing$}
    \For{$i \in \{1, \ldots, n\}$, $j \in \{1, \ldots, m\}$}
        \State {\textbf{let} $ineqs \leftarrow \varnothing$ }
        \State {\textbf{let} $consequent \leftarrow \mathrm{\bf copy} \ \mathrm{ of } \ \varphi_i \leq \psi_j$}
        \For{every maximal PIA subformula $\alpha$ of $\varphi_i \leq \psi_j$ (see note below)}
            \State{\textbf{let} $a$ be a fresh variable}
            \State{$ineqs \leftarrow ineqs \cup \{a \leq \alpha \}$ {\bf if} $\alpha$ is positive in $\varphi_i \leq \psi_j$ \\ \hspace*{\algorithmicindent}\hspace*{\algorithmicindent}\hspace*{\algorithmicindent}\phantom{$ineqs \leftarrow$} $ineqs \cup \{\alpha \leq a \}$ {\bf otherwise}}
            \State{\textbf{replace} $\alpha$ in $consequent$ by $a$}
        \EndFor
        \State {$output \leftarrow output \cup \{ (ineqs, consequent) \}$}
    \EndFor
    \State\Return{$output$}
\EndFunction
\end{algorithmic}
\hspace*{\algorithmicindent} {\footnotesize \textbf{Note}: maximal PIA subformulas can be easily enumerated when the input inequality $\varphi\leq\psi$ is encoded by its syntax tree: it is sufficient to perform a DFS visit of the two sides of the inequality which keeps track of the polarity of the node currently visited. As soon as a PIA node is encountered ($+g$ or $-f$), it is returned, and the visit stops proceeding more in depth from that node. For certain applications, it might be useful to encode analytic inequalities by storing their skeleton and the list of maximal PIA subformulas, in which case, it would be even more straightforward to enumerate them.}
\end{algorithm}

\begin{algorithm}
\caption{Algorithm $\mathsf{FLAT}$}
\label{algo:flat}
\small
\hspace*{\algorithmicindent} \textbf{Input}: An inequality $\varphi \leq \psi$, such that $\varphi$ and $\psi$ only contain skeleton nodes. \\
\hspace*{\algorithmicindent} \textbf{Output}: a set of inequalities equivalent to $\varphi \leq \psi$, and such that each element contains almost one connective.
\begin{algorithmic}[1]
\Function{{Flat}}{$\varphi$, $\psi$}
    \If{$\varphi\leq\psi$ contains at most one connective}
        \State\Return{$\{ \varphi\leq\psi \}$}
    \EndIf
    \If{$\varphi$ is a variable}
        \State{{\bf let} $\psi_1,\ldots,\psi_n$ be the children of $\psi$'s root connective}
        \If{$\psi$'s root connective is $\wedge$}
            \State\Return{\textsc{Flat}$(\{\varphi \leq \psi_1\}) \ \cup $ \textsc{Flat}$(\{\varphi \leq \psi_2\})$}
        \Else
            \State {\textbf{let} $h$ be $\psi$'s root connective, and $c_1,\ldots,c_n$ be fresh variables}
            \State\Return{$\{ \varphi \leq h(c_1,\ldots,c_n) \} \ \cup \bigcup_i^n$ \textsc{Flat}$(\{c_i \leq^{\varepsilon_h(i)} \psi_i\})$ }
        \EndIf
    \Else
        \State{{\bf let} $\varphi_1,\ldots,\varphi_n$ be the children of $\varphi$'s root connective}
        \If{$\varphi$'s root connective is $\vee$}
            \State\Return{\textsc{Flat}$(\{\varphi_1 \leq \psi\}) \ \cup $ \textsc{Flat}$(\{\varphi_2 \leq \psi\})$}
        \Else
            \State {\textbf{let} $h$ be $\varphi$'s root connective, and $c_1,\ldots,c_n$ be fresh variables}
            \State\Return{$\{ h(c_1,\ldots,c_n) \leq \psi \} \ \cup \bigcup_i^n$ \textsc{Flat}$(\{\varphi_i \leq^{\varepsilon_h(i)} c_i\})$ }
        \EndIf
    \EndIf
\EndFunction
\end{algorithmic}
\hspace*{\algorithmicindent} \textbf{Input}: A finite collection $X$ of quasi-inequalities $\bigwith_i \varphi_i\leq \psi_i\Rightarrow \varphi\leq \psi$ s.t.~$\varphi_i, \varphi$ are positive Skeleton and $\psi_i, \psi$ are negative Skeleton. \\
\hspace*{\algorithmicindent} \textbf{Output}: A finite collection of formulas obtained by applying \textsc{Flat} to all inequalities in each formula in $X$.
\begin{algorithmic}[1]
\Function{$\mathsf{FLAT}$}{$X$}
    \For{all inequalities $\varphi\leq\psi$ in formulas $\xi$ in $X$}
        \State {\bf substitute} the occurrence of the inequality with $\bigwith$ \textsc{Flat}$(\varphi, \psi)$
    \EndFor
    \State\Return $X$
\EndFunction
\end{algorithmic}
\end{algorithm}

\section{Example}
\label{sec:examples}
Let us demonstrate the procedure described in the previous section on the $\mathcal{L}$-inequality \[\tw (\blacksquare \wt\tw \tw \wt v_1\wedge \wt \Diamond \wt v_1)\leq \blacksquare \wt (\bt v_1\vee \blacksquare \bt (\tb \wt v_2\vee \Diamond v_2 ))\]
which was shown to be  analytic   in Example \ref{ex:running example}.

In step 1, we distribute the leftmost occurrence of $\tw$ over $\wedge$ on the lefthand side of the inequality, and the  terms $\blacksquare \wt$ and $\blacksquare \bt$ over $\vee$, thus obtaining

\[ \tw\blacksquare \wt\tw \tw \wt v_1\vee \tw\wt \Diamond \wt v_1\leq  \blacksquare \wt\bt v_1\wedge \blacksquare \wt(\blacksquare \bt \tb \wt v_2\wedge \blacksquare \bt \Diamond v_2 ).\]
By exhaustively applying the splitting rules, we obtain the following four definite analytic inequalities, on each of which we proceed separately:
\begin{center}
\begin{tabular}{ll}
$\tw\blacksquare \wt\tw \tw \wt v_1 \leq \blacksquare \wt\bt v_1 $& $  \tw\wt \Diamond \wt v_1 \leq \blacksquare \wt\bt v_1$\\

$\tw\blacksquare \wt\tw \tw \wt v_1 \leq  \blacksquare \wt(\blacksquare \bt \tb \wt v_2\wedge \blacksquare \bt \Diamond v_2 )$& $

\tw\wt \Diamond \wt v_1 \leq  \blacksquare \wt(\blacksquare \bt \tb \wt v_2\wedge \blacksquare \bt \Diamond v_2 ).$\\
\end{tabular}
\end{center}
In the remainder of this section, we apply the procedure  only on the lower left inequality above.  This is a definite analytic inequality $\varphi [\beta/!y] \leq \psi [\alpha/!x]$, with $\varphi(!y)\coloneqq \tw\blacksquare \wt\tw y$ definite positive Skeleton, $\psi(!x)\coloneqq \blacksquare \wt x$ definite negative Skeleton, $\alpha (v_2)\coloneqq \gamma_1\wedge \gamma_2$ where $\gamma_1\coloneqq \blacksquare \bt \tb \wt v_2$ and $\gamma_2\coloneqq \blacksquare \bt \Diamond v_2 $ are definite positive PIA  and $\beta =  \delta_1\coloneqq \tw\wt v_1$ is definite negative PIA.
In step 2, using denseness, we introduce fresh open and closed variables and equivalently rewrite the lower left inequality above as the following quasi-inequality:
\[\forall k \forall o(\tw \wt v_1\leq o\ \&\ k\leq \blacksquare \bt \tb \wt v_2\wedge \blacksquare \bt \Diamond v_2 \Rightarrow \tw\blacksquare \wt\tw o \leq  \blacksquare \wt k).\]
In step 3, we apply compactness and introduce two fresh variables ranging in $A$:
\[\forall k \forall o\forall a \forall b(\tw \wt v_1\leq b\leq o\ \&\ k\leq a\leq  \blacksquare \bt \tb \wt v_2\wedge \blacksquare \bt \Diamond v_2 \Rightarrow \tw\blacksquare \wt\tw o \leq  \blacksquare \wt k).\]
In step 4, we apply Ackermann's lemma to eliminate $k$ and $o$:
\[\forall a \forall b(\tw \wt d\leq b \ \&\  a\leq  \blacksquare \bt \tb \wt v_2\wedge \blacksquare \bt \Diamond v_2 \Rightarrow \tw\blacksquare \wt\tw b \leq  \blacksquare \wt a).\]
In step 5, we
first apply compactness to the inequality in the consequent of the quasi-inequality above, and splitting to eliminate $\wedge$ in the second inequality in the antecedent, thus obtaining:
\[\forall a \forall b\exists d_0(\tw \wt v_1\leq b \ \&\  a\leq  \blacksquare \bt \tb \wt v_2\ \&\  a\leq \blacksquare \bt \Diamond v_2 \Rightarrow \tw\blacksquare \wt\tw b \leq d_0\ \&\ d_0\leq   \blacksquare \wt a).\]
In the quasi inequality above, every inequality is of the syntactic shape required for applying Lemma \ref{lemma: flattification of definite skeleton into chains}. Thus, we rewrite each inequality into a chain of flat inequalities by  recursively applying the procedure indicated in the proof of Lemma \ref{lemma: flattification of definite skeleton into chains}:
\begin{center}
\begin{tabular}{r c l}
$\tw\wt v_1\leq b$ &iff & $\exists c_1(\tw c_1\leq b\ \&\ c_1\leq \wt v_1)=: \Xi_{\delta_1}(b, v_1)$\\
& \\
$a\leq  \blacksquare \bt \tb \wt v_2$ & iff & $\exists c_2(a\leq \blacksquare c_2\ \&\ c_2\leq \bt \tb \wt v_2)$\\
&iff & $\exists c_2\exists c_3(a\leq \blacksquare c_2\ \&\ c_2\leq \bt c_3\ \& \ \tb \wt v_2\leq c_3)$\\
&iff & $\exists c_2\exists c_3\exists c_4(a\leq \blacksquare c_2\ \&\ c_2\leq \bt c_3\ \& \ \tb c_4 \leq c_3\ \&\ c_4\leq \wt v_2)=: \Xi_{\gamma_1}(a,  v_2)$\\
&\\
$a\leq \blacksquare \bt \Diamond v_2 $ &iff &$\exists c_5(a\leq \blacksquare c_5\ \&\ c_5\leq \bt \Diamond v_2)$\\
&iff &$\exists c_5\exists c_6(a\leq \blacksquare c_5\ \&\ c_5\leq \bt c_6 \ \&\ \Diamond v_2\leq c_6)=:\Xi_{\gamma_2}(a,  v_2)$\\
&\\
$\tw\blacksquare \wt\tw b \leq d_0$ & iff & $\exists d_1(\tw d_1 \leq d_0\ \&\ d_1\leq \blacksquare \wt\tw b)$\\
& iff & $\exists d_1\exists d_2(\tw d_1 \leq d_0\ \&\ d_1\leq \blacksquare d_2\ \&\ d_2\leq  \wt\tw b)$\\
& iff & $\exists d_1\exists d_2\exists d_3(\tw d_1 \leq d_0\ \&\ d_1\leq \blacksquare d_2\ \&\ d_2\leq  \wt d_3 \ \&\ \tw b\leq d_3)=: \Xi_{\varphi}(d_0, b)$\\
&\\
$d_0\leq   \blacksquare \wt a$ & iff & $\exists d_4(d_0\leq   \blacksquare d_4\ \&\ d_4\leq  \wt a)=:\Xi_{\psi}(d_0,  a)$
\end{tabular}
\end{center}
Hence we obtain:
\[\forall \overline{v}\forall a \forall b\forall \overline{c}\exists d_0\exists\overline{d}(\Xi'_{\delta_1}(b, \overline{c}, v_1) \ \&\  \Xi'_{\gamma_1}(a, \overline{c}, v_2)\ \&\  \Xi'_{\gamma_2}(a, \overline{c}, v_2) \Rightarrow \Xi'_{\varphi}(d_0, \overline{d}, b)\ \&\ \Xi'_{\psi}(d_0,\overline{d}, a)).\]
In step 6, we  translate the flat inequalities in each $\Xi'$ so as to obtain the following formula in the first order language of spd-algebras:
\[\forall \overline{v}\forall a \forall b\forall \overline{c}\exists d_0\exists\overline{d}(\Xi^\tau_{\delta_1}(b, \overline{c}, v_1) \ \&\  \Xi^\tau_{\gamma_1}(a, \overline{c}, v_2)\ \&\  \Xi^\tau_{\gamma_2}(a, \overline{c}, v_2) \Rightarrow \Xi^\tau_{\varphi}(d_0, \overline{d}, b)\ \&\ \Xi^\tau_{\psi}(d_0,\overline{d}, a)).\]
where:
\begin{center}
    \begin{tabular}{lllll}
    $\Xi^\tau_{\delta_1}(b, \overline{c}, v_1)$ & = & $\tau(\tw c_1\leq b)\ \&\  \tau(c_1\leq \wt v_1)$ \\
    &  = & $b\nppcon c_1 \ \&\ v_1\npcon c_1 $\\
         &  \\
      $\Xi^\tau_{\gamma_1}(a, \overline{c}, v_2)$ & = &  $\tau(a\leq \blacksquare c_2)\ \&\ \tau(c_2\leq \bt c_3)\ \& \ \tau(\tb c_4 \leq c_3)\ \&\ \tau(c_4\leq \wt v_2) $\\
      & = & $a\prec c_2 \ \&\ c_2\npcon c_3\ \&\ c_4\nppcon c_3\ \&\ v_2\npcon c_4 $\\
      &\\
    $\Xi^\tau_{\gamma_2}(a, \overline{c}, v_2)$ & = & $\tau(a\leq \blacksquare c_5)\ \&\ \tau(c_5\leq \bt c_6) \ \&\ \tau(\Diamond v_2\leq c_6)$\\
   & = & $ a\prec c_5\ \&\ c_5\npcon c_6\ \&\ v_2\prec c_6$\\
      &\\
$\Xi^\tau_{\varphi}(d_0,\overline{d}, b)$&= &$\tau(\tw d_1 \leq d_0)\ \&\ \tau(d_1\leq \blacksquare d_2)\ \&\ \tau(d_2\leq  \wt d_3) \ \&\ \tau(\tw b\leq d_3)$  \\
& = & $d_0\npcon d_1\ \&\ d_1\prec v_2\ \&\ v_3\npcon d_2\ \&\ d_3\nppcon b$\\
&\\
$\Xi^\tau _{\psi}(d_0, \overline{d}, a)$ & = &$\tau(d_0\leq   \blacksquare d_4)\ \&\ \tau(d_4\leq  \wt a)$\\
& = & $d_0\prec d_4\ \&\ a\npcon d_4$.\\
    \end{tabular}
\end{center}

\section{Kracht formulas}
\label{sec:kracht_formulas}

In what follows, for any $\mathcal{L}^{\mathrm{FO}}_{spd}$-atom $sRt$,  we let $\rho(sRt)$ denote one of the two inequalities which translate to $sRt$
 as indicated in Step 6 of Section \ref{sec:correspondence}. Since the two inequalities associated with $sRt$ are equivalent modulo residuation, and since residuation preserves the polarity of variable occurrences, we can say that a variable occurrence in an $\mathcal{L}^{\mathrm{FO}}_{spd}$-atom $sRt$ is {\em positive} (resp.~{\em negative}) if it is positive (resp.~negative) in the signed generation trees of  $\rho(sRt)$. Moreover, if all nodes of the branch from that occurrence to the root are Skeleton nodes (cf.~Section \ref{ssec:analytic-LE-axioms}), we will also say that the given occurrence is {\em displayable}.\footnote{Indeed, the next lemma shows that if a variable occurrence is displayable, then, modulo application of splitting rules and residuation rules (cf.~Footnote \ref{ftn:display or residuation}),  it can be displayed. However, notice that this fact only holds under the   additional requirement that, for every SLR node  occurring in the branch from $x$ to the root,  the signature $\mathcal{L}$ includes the residual of the corresponding connective  in the coordinate through which the branch from $x$ to the root passes.}

\begin{lemma}
\label{lemma:displayable}
 Any   $\mathcal{L}_\mathrm{LE}$-inequality $(\phi\leq \psi)[v/!x]$ such that the placeholder variable  $x$ is displayable can  equivalently be rewritten as a conjunction of inequalities, one of which is  of the form $v \leq \theta$ for some $\mathcal{L}_\mathrm{LE}$-formula $\theta$, if $x$ is positive, or of the form $\theta\leq v$, if $x$ is negative.
\end{lemma}
\begin{proof}
By induction on the number $n$ of ancestors of the displayable occurrence.
If $n = 0$, then the inequality is either of the form $x\leq \psi$ or of the form $\varphi \leq x$, hence the claim is immediately true. As to the inductive step, we proceed by cases.
If $\varphi\coloneqq \varphi'(!x)\vee \varphi''$,  
then $\varphi \leq \psi$ is equivalent to  the conjunction of  the inequalities $\varphi'(!x)\leq \psi$ and $ \varphi''\leq \psi$. 
By applying the inductive hypothesis, $\varphi'(!x)\leq \psi$  can be equivalently rewritten as a conjunction of inequalities, one of which of the form $v\leq \theta$ or $\theta\leq v$. Hence, by adding $ \varphi''\leq \psi$ to this conjunction, we obtain the required conjunction of inequalities. The proof for $\psi\coloneqq  \psi'(!x)\wedge\psi''$ is analogous.

If $\varphi\coloneqq f (\overline{\chi}/!\overline{z})$ with   $\chi_i = \chi(!x)$ and $\varepsilon_f(i) = 1$, then, by applying residuation on the $i$th coordinate of $f$ (cf.~Footnote \ref{ftn:display or residuation}), $\varphi\leq \psi$ is equivalent to $\chi(!x)\leq f^\sharp_i(\psi/! z_i, \overline{\chi'}/!\overline{z'})$, where $\overline{\chi'}$ denotes the vector $\overline{\chi}$ except its $i$th coordinate. Then the required conjunction is obtained by applying the induction hypothesis to this inequality. The proofs  for $\varepsilon_f(i) = \partial$ and $\psi \coloneqq g(\overline{\chi}/!\overline z)$ are analogous.
\end{proof}

The following abbreviations will be used throughout the present section:

{{\centering
\begin{tabular}{lclcl}
$(\exists y \succ x)\varphi$ & $\equiv$ & $\exists y(x \prec y \ \& \ \varphi)$ &
i.e.& $\exists y(x \leq \blacksquare y \ \& \ \varphi)$ \\
$(\exists y \prec x)\varphi$ & $\equiv$ & $\exists y(y \prec x \ \& \ \varphi)$ &
i.e. &$\exists y(\Diamond y \leq x \ \& \ \varphi)$ \\
$(\exists y \npcon x)\varphi$ & $\equiv$ & $\exists y(y \npcon x \ \& \ \varphi)$ &
i.e.& $\exists y(x \leq {\rhd} y \ \& \ \varphi)$ \\
$(\exists y \npcon^{-1} x)\varphi$ & $\equiv$ & $\exists y(x \npcon y \ \& \ \varphi)$ &
i.e. &$\exists y(x \leq {\blacktriangleright} y \ \& \ \varphi)$ \\
$(\exists y \nppcon x)\varphi$ & $\equiv$ & $\exists y(y \nppcon x \ \& \ \varphi)$ &
i.e. &$\exists y({\blacktriangleleft} y \leq x \ \& \ \varphi)$ \\
$(\exists y \nppcon^{-1} x)\varphi$ & $\equiv$ & $\exists y(x \nppcon y \ \& \ \varphi)$ &
i.e. & $\exists y({\lhd} y \leq x \ \& \ \varphi)$\\
$(\exists \overline{y} \leq_f x)\varphi$ & $\equiv$ & $\exists \overline y(f(\overline y) \leq x \ \& \ \varphi)$\\
$(\exists \overline{y} \geq_g x)\varphi$ & $\equiv$ & $\exists \overline y(x \leq g(\overline y) \ \& \ \varphi)$
\end{tabular}
\begin{tabular}{lclcl}
$(\forall y \succ x)\varphi$ & $\equiv$ & $\forall y(x \prec y \ \Rightarrow \ \varphi)$ &
i.e.& $\forall y(x \leq \blacksquare y \ \Rightarrow \ \varphi)$ \\
$(\forall y \prec x)\varphi$ & $\equiv$ & $\forall y(y \prec x \ \Rightarrow \ \varphi)$ &
i.e. &$\forall y(\Diamond y \leq x \ \Rightarrow \ \varphi)$ \\
$(\forall y \npcon x)\varphi$ & $\equiv$ & $\forall y(y \npcon x \ \Rightarrow \ \varphi)$ &
i.e.& $\forall y(x \leq {\rhd} y \ \Rightarrow \ \varphi)$ \\
$(\forall y \npcon^{-1} x)\varphi$ & $\equiv$ & $\forall y(x \npcon y \ \Rightarrow \ \varphi)$ &
i.e. &$\forall y(x \leq {\blacktriangleright} y \ \Rightarrow \ \varphi)$ \\
$(\forall y \nppcon x)\varphi$ & $\equiv$ & $\forall y(y \nppcon x \ \Rightarrow \ \varphi)$ &
i.e. &$\forall y({\blacktriangleleft} y \leq x \ \Rightarrow \ \varphi)$ \\
$(\forall y \nppcon^{-1} x)\varphi$ & $\equiv$ & $\forall y(x \nppcon y \ \Rightarrow \ \varphi)$ &
i.e. & $\forall y({\lhd} y \leq x \ \Rightarrow \ \varphi)$\\
$(\forall \overline{y} \leq_f x)\varphi$ & $\equiv$ & $\forall \overline y(f(\overline y) \leq x \Rightarrow \varphi)$\\
$(\forall \overline{y} \geq_g x)\varphi$ & $\equiv$ & $\forall \overline y(x \leq g(\overline y) \Rightarrow \varphi)$ \\
\end{tabular}
\par}}
\smallskip

\noindent where $f \in \mathcal{F}_A$ and $g \in \mathcal{G}_A$. Expressions  such as $(\forall y\prec x)$ or $(\exists\overline y\geq_g x)$ above are referred to as {\em restricted quantifiers}.

The variable $y$ (resp.~the variables in $\overline y$)  in the formulas above is (resp.\ are) {\em restricted}, and the variable $x$ is {\em restricting}, while the inequality occurring together with $\varphi$ in the translation of the restricted quantifier is a {\em restricting inequality}. 

Throughout the present section, we will use the following letters (possibly with sub- or superscripts) to distinguish the roles of different variables ranging in the domain of an arbitrary spd-algebra, both in respect to  Algorithms \ref{algo:alba}  and \ref{algo:flat} of Section \ref{sec:correspondence} and to the ones of Section \ref{sec:inverse_correspondence_procedure}.  These variables will be assigned different conditions in Definition \ref{def:inverse_shape} on   their polarity and their distribution in the formula, which determines the way in which they are introduced/eliminated:

\begin{itemize}
\item[$v$] variables occurring in the input inequality of Algorithm \ref{algo:alba}, or in the output inequalities of Algorithm \ref{algo:inverse};
\item[$a$] positive variables introduced/eliminated by means of the first approximation lemma (cf.\ Lemma \ref{lemma:first_approx}). It must be possible to rewrite the quasi-inequality so that each such variable  occurs only once on each side of the main implication;
\item[$b$] negative variables introduced/eliminated by means of the first approximation lemma. The same considerations which apply to $a$-variables apply also to $b$-variables;
\item[$c$] variables introduced/eliminated  using Lemma \ref{lem: diamond-output equivalence extended} in the antecedent of the main implication;
\item[$d$] variables introduced/eliminated using Lemma \ref{lem: diamond-output equivalence extended} in the consequent of the main implication.
\end{itemize}

\begin{definition}
\label{def:inverse_shape}
An $\mathcal{L}^{\mathrm{FO}}_{spd}$-sentence is a {\em Kracht formula}\footnote{In the modal logic literature, Kracht formulas (cf.~\cite[Section 3.7]{blackburn2001modal}) are those first order sentences in the frame correspondence language of Kripke frames, introduced by Kracht \cite{kracht1999tools}, each of which is (equivalent to) the first-order correspondent of some Sahlqvist axiom. The notion of Kracht formulas has been recently generalized in \cite{palmigiano2024unified, conradie2022unified} from classical modal logic to (distributive) LE-logics, and from a class of first order formulas targeting Sahlqvist LE-axioms to a class targeting the strictly larger class of inductive LE-axioms. The notion of Kracht formulas introduced in Definition \ref{def:inverse_shape} is different and in fact incomparable with those in \cite{palmigiano2024unified,conradie2022unified}, since it targets a different and incomparable class of modal axioms.} if it has the following shape:
\[
\forall \overline v\forall \overline a, \overline b
(\forall \overline{c_m}R'_m z_m)\cdots(\forall \overline{c_1} R'_1 z_1)
\left(
\eta
\Rightarrow
(\exists \overline{d_o} R_o y_o)\cdots(\exists \overline{d_1} R_1 y_1)
\zeta
\right)
\]
where
\begin{enumerate}
\item $R'_i, R_j\in \{\leq, \geq, \prec, \succ, \npcon, \npcon^{-1}, \nppcon, \nppcon^{-1}, \leq_f, \geq_g \}$ for all $1\leq i\leq m$ and $1\leq j\leq o$;
\item variables in $\overline z$ are amongst those of $\overline a$, $\overline b$, and $\overline c \coloneqq \overline{c_1}\oplus\ldots \oplus\overline{c_m}$;\footnote{Here $\oplus$ denotes the concatenation of sequences, i.e. $(a_1, \ldots, a_n) \oplus (b_1, \ldots, b_m) \coloneqq (a_1, \ldots, a_n, b_1, \ldots, b_m)$.} variables in $\overline y$ are amongst those of $\overline a$, $\overline b$, and $\overline d \coloneqq \overline{d_1}\oplus\ldots \oplus\overline{d_m}$;
\item $\eta$ and $\zeta$ are conjunctions of relational atoms $sRt$ with $R\in \{ \leq, \geq, \prec, \succ, \npcon, \npcon^{-1}, \nppcon, \nppcon^{-1} \}$;
\item all occurrences of variables in $\overline a$ (resp.\ $\overline b$) 
are positive (resp.\ negative) in all atoms (including those in restricting quantifiers) in which they occur;

\item every variable in $\overline c$ (resp.\ in $\overline d$) occurs uniformly in $\eta$ (resp.~in $\zeta$);

\item occurrences of variables in $\overline c$ (resp.\ $\overline d$) as restricting variables have the same polarity as their occurrences in $\eta$ (resp.\ $\zeta$).

\item all occurrences of variables in $\overline a$, $\overline b$, $\overline c$, and $\overline d$ in atoms of  $\eta$ and $\zeta$ are displayable;

\item each atom in $\eta$ contains exactly one occurrence of a variable in $\overline a$, $\overline b$, or $\overline c$ (all other variables are in $\overline v$);

\item each atom $sRt$ in $\zeta$ contains at most one occurrence of a variable in $\overline d$. Moreover, for every two different occurrences of the same variable not in $\overline v$ (i.e., in $\overline a$ or $\overline b$), the first common ancestor in the signed generation tree of $sRt$ is either $+\wedge$ or $-\vee$.
\end{enumerate}
\end{definition}

\begin{remark}
\label{remark:comparison_with_corr_output_shape}
The definition above captures a more general syntactic shape of $\mathcal{L}^{\mathrm{FO}}_{spd}$-formulas than the general shape of the output of the algorithm defined in Section \ref{sec:correspondence}. Indeed, such outputs have a much simpler antecedent, composed only of atoms $xRt$ and $tRx$, with $x$ in $\overline c$ or in $\overline a$ if in positive polarity and $y$ in $\overline b$ if in negative polarity, and with $t$ containing only variables in $\overline v$. This automatically simplifies some of the items of Definition \ref{def:inverse_shape}, such as item 9. Similar considerations also apply to the atoms in the consequent.
\end{remark}

\begin{example}
The  condition $\forall v\forall a \forall c (a \leq c \ \& \ c \prec v  \Rightarrow a \npcon v)$ is a Kracht formula, since it can be rewritten as $\forall v\forall a(\forall c \geq a)(c \prec v \ \Rightarrow a \npcon v)$, and in this shape, it readily satisfies conditions 1-9 of Definition  \ref{def:inverse_shape}. The following computation illustrates how  Algorithm \ref{algo:inverse} (cf.~Section \ref{sec:inverse_correspondence_procedure}) runs on it, and shows how an equivalent modal axiom can be effectively computed by eliminating the variables $c$ and $a$.

\smallskip

{{\centering
\begin{tabular}{rl}
 & $\forall v\forall a \forall c (a \leq c \ \& \ c \prec v  \Rightarrow a \npcon v)$ \\
iff & $\forall v\forall a \forall c (a \leq c \ \& \ c \leq \blacksquare v\Rightarrow a \leq \bt v)$ \\
iff & $\forall v\forall a (\exists c(a \leq c \ \& \ c \leq \blacksquare v) \Rightarrow a \leq \bt v)$ \\
iff & $\forall v\forall a (a \leq \blacksquare v  \Rightarrow a \leq \bt v)$ \\
iff & $\forall v (\blacksquare v  \leq \bt v)$.\\
\end{tabular}
\par}}
\end{example}
\begin{example}
Unrestricted variables in $\overline c$ (resp.\ $\overline d$) are allowed in Definition \ref{def:inverse_shape} because they can be readily eliminated, given that they occur uniformly in $\eta$ (resp.\ $\zeta$). To illustrate how this elimination is effected, consider the condition $\forall v \forall a \forall c_1 \forall c_2 (a\prec c_2 \ \&\ c_2\prec v\ \&\  c_1 \prec v \Rightarrow a \prec v)$, which is a Kracht formula when equivalently rewritten as $\forall v \forall a \forall c_1 (\forall c_2 \succ a)( c_2\prec v \ \&\ c_1 \prec v \Rightarrow a \prec v)$. Applying Algorithm \ref{algo:inverse} to it yields:
\smallskip

{{\centering
\begin{tabular}{rll}
& $\forall v \forall a \forall c_1 \forall c_2  (a\prec c_2 \ \&\ c_2\prec v\ \&\  c_1 \prec v \Rightarrow a \prec v)$ \\
iff & $\forall v \forall a \forall c_1 \forall c_2  (a\leq \blacksquare c_2\ \&\  c_2 \leq \blacksquare v \ \&\  c_1 \leq \blacksquare v\Rightarrow a \leq \blacksquare v)$ \\
iff & $\forall v \forall a  (\exists c_2 (a\leq \blacksquare c_2 \ \&\  c_2 \leq \blacksquare v)\ \&\  \exists c_1(c_1 \leq \blacksquare v) \Rightarrow a \leq \blacksquare v)$ \\
iff & $\forall v \forall a  (a\leq \blacksquare \blacksquare v \Rightarrow a \leq \blacksquare v)$ & ($\ast$) \\
iff & $\forall v (\blacksquare \blacksquare v \leq \blacksquare v)$. \\
\end{tabular}
\par}}
\noindent While, in the equivalence marked with ($\ast$), the soundness of the elimination of $c_2$ is a consequence of Lemma \ref{lem: diamond-output equivalence extended}.2(i), the soundness of the elimination of $c_1$ is a consequence of the fact that, on bounded lattices, the formula $\exists c_1(c_1 \leq \blacksquare v)$ is always true. The soundness of the elimination of $a$ in the last line is due to the fact that, by Lemma \ref{lemma: distribution properties of sigma pi}.2, both $\blacksquare v$ and $\blacksquare\blacksquare v$ belong to $O(A^\delta)$.
\end{example}

\begin{example}
Consider the universal closure of the following condition  \cite{Makinson03}:
\[
\mathrm{(OR)}^\downarrow \quad\quad
a \prec x\ \&\ b \pcon  x \Rightarrow (a \vee b) \pcon x
\]
The quasi inequality above satisfies  Definition \ref{def:inverse_shape} when $a$ is regarded as a $\overline c$-variable, $b$ as an $\overline a$-variable, and $x$ as a $\overline v$-variable. Applying Algorithm \ref{algo:inverse} to it yields:

{{\centering
\begin{tabular}{rll}
     &  $a \prec x\ \&\ b \pcon  x \Rightarrow (a \vee b) \pcon x$ \\
iff &  $a \prec x\ \&\ (a \vee b) \npcon x \Rightarrow b \npcon  x$ \\
iff &  $a \leq \blacksquare x\ \&\ a \vee b \leq {\blacktriangleright} x \Rightarrow b \leq {\blacktriangleright}  x$ \\
iff &  $a \leq \blacksquare x\ \&\ a \leq {\blacktriangleright} x \ \& \ b \leq {\blacktriangleright} x \Rightarrow b \leq {\blacktriangleright}  x$ & ($\ast \ast$)  \\
iff &  $a \leq \blacksquare x \wedge {\blacktriangleright} x \ \& \ b \leq {\blacktriangleright} x \Rightarrow b \leq {\blacktriangleright}  x$ \\
iff &  $\exists a (a \leq \blacksquare x \wedge {\blacktriangleright} x) \ \& \ b \leq {\blacktriangleright} x \Rightarrow b \leq {\blacktriangleright}  x$ \\
iff &  $b \leq {\blacktriangleright} x \Rightarrow b \leq {\blacktriangleright}  x$ & ($\ast$) \\
iff &  ${\blacktriangleright} x \leq {\blacktriangleright}  x$. \\
\end{tabular}
\par}}
\smallskip
The soundness of the elimination of $a$ is a consequence of the fact that, on bounded lattices, the formula $\exists a(a \leq \blacksquare x \wedge {\blacktriangleright} x)$ is always true, while the soundness of the elimination of $b$ is due to the fact that $\bt x\in O(A^\delta)$. It turns out that, in the present setting, this condition is a tautology. This is unsurprising, since, by definition, $\pcon$ satisfies condition (C3) in Section \ref{ssec:subordination precontact}, which is a necessary condition for the modal operators $\bt$ and $\wt$ to be normal (i.e.~finitely join-reversing).
\end{example}

\begin{example}
Consider the universal closure of the following condition  from \cite{Makinson03} on any spd-algebra $\mathbb{H}$ based on a Heyting algebra:

\[\mathrm{(CTA)}^{\downarrow} \quad\quad z\vee y \prec \neg x\ \&\ (z \wedge x)\, \pcon \, y \Rightarrow z  \pcon  (x\wedge y).\]
The quasi inequality above satisfies  Definition \ref{def:inverse_shape} when any  variable is picked to play the role of an $\overline a$-variable, and the remaining variables  as $\overline v$-variables. By choosing $x$ as the $\overline a$-variable, and applying Algorithm \ref{algo:inverse} to it, we get:

{{\centering
\begin{tabular}{rl}
&  $z\vee y \prec \neg x\ \&\ (z \wedge x)\, \pcon \, y \Rightarrow z  \pcon  (x\wedge y)$ \\
iff &  $z\vee y \prec \neg x\ \&\ z  \npcon  (x\wedge y) \Rightarrow (z \wedge x)\, \npcon \, y$ \\
iff &  $\Diamond(z\vee y)\leq \neg x\ \&\  x\wedge y\leq {\rhd} z\Rightarrow z\wedge x\leq {\blacktriangleright} y$ \\
iff & $x\leq \neg\Diamond(z\vee y)\ \&\  x\leq y\to {\rhd} z\Rightarrow x\leq z\to {\blacktriangleright} y$\\
iff & $\neg\Diamond(z\vee y)\wedge (y\to {\rhd} z)\leq z\to {\blacktriangleright} y$
\end{tabular}
\par}}
which shows that \[\mathbb{H}\models \mathrm{(CTA)}^{\downarrow} \quad \text{ iff }\quad \mathbb{H}^\ast\models \neg\Diamond(z\vee y)\wedge (y\to {\rhd} z)\leq z\to {\blacktriangleright} y.\]
\end{example}

\begin{example}
Consider the universal closure of the following condition from \cite{parent2019input}:
\[\mathrm{(MCT)} \quad\quad  a\prec x' \ \&\ x'\leq x \ \&\ a\wedge x\prec y\Rightarrow a\prec y.\]
The quasi inequality above satisfies  Definition \ref{def:inverse_shape}, since it can be equivalently rewritten as
$
\forall a \forall y(\forall x' \succ a)(\forall x \geq x')(a \wedge x \prec y \Rightarrow a \prec y)
$ when $a$ is regarded as a $\overline v$-variable, $y$ as a $\overline b$-variable, and $x$ and $x'$ as  $\overline c$-variables. Applying Algorithm \ref{algo:inverse} to it yields:

{{\centering
\begin{tabular}{rl}
& $\forall a \forall y\forall x' \forall x( a\prec x'\ \&\ x'\leq x \ \& \ a \wedge x \prec y \Rightarrow a \prec y)$ \\
iff & $\forall a \forall y\forall x' \forall x(\Diamond a \leq x'  \ \& \ x' \leq x  \ \& \ \Diamond (a \wedge x) \leq  y\Rightarrow \Diamond a \leq y)$ \\
iff & $\forall a \forall y\forall x' (\Diamond a \leq x'  \ \& \ \exists x(x' \leq x \ \& \ \Diamond (a \wedge x) \leq  y)\Rightarrow \Diamond a \leq y)$ \\
iff & $\forall a \forall y\forall x' (\Diamond a \leq x'  \ \& \ \Diamond (a \wedge x') \leq  y\Rightarrow \Diamond a \leq y)$ \\
iff & $\forall a \forall y(\exists x'(\Diamond a \leq x' \ \& \ \Diamond (a \wedge x')) \leq  y \Rightarrow \Diamond a \leq y)$ \\
iff & $\forall a \forall y(\Diamond (a \wedge \Diamond a) \leq  y \Rightarrow \Diamond a \leq y)$ \\
iff & $\forall a(\Diamond a \leq \Diamond(a \wedge \Diamond a))$ \\
\end{tabular}
\par}}
\end{example}

\begin{example}
Consider the universal closure of the following condition from \cite{parent2019input} on any spd-algebra $\mathbb{H}$ based on a Heyting algebra:
\[\mathrm{(ACT)} \quad\quad a\prec x\ \&\ a\wedge x\prec y \Rightarrow a\prec x \wedge y\]

The quasi inequality above satisfies  Definition \ref{def:inverse_shape} when $a$ is regarded as  an $\overline a$-variable, and the remaining variables  as $\overline v$-variables. Applying Algorithm \ref{algo:inverse} to it yields:

{{\centering
\begin{tabular}{rl}
&$\forall x \forall y \forall a (a\prec x\ \&\ a\wedge x\prec y \Rightarrow a\prec x \wedge y) $\\
iff & $\forall x \forall y \forall a (
a \leq \blacksquare x \ \& \ a \wedge x \leq \blacksquare y \Rightarrow a \leq \blacksquare(x \wedge y))$ \\
iff &  $\forall x \forall y \forall a (
a \leq \blacksquare x \ \& \ a \leq x \rightarrow \blacksquare y \Rightarrow a \leq \blacksquare(x \wedge y))$ \\
iff &  $\forall x \forall y \forall a (
a \leq \blacksquare x \wedge (x \rightarrow \blacksquare y) \Rightarrow a \leq \blacksquare(x \wedge y))$ \\
iff &  $\forall x \forall y (\blacksquare x \wedge (x \rightarrow \blacksquare y) \leq \blacksquare(x \wedge y))$
\end{tabular}
\par}}

\end{example}

\begin{example}
Consider the universal closure of the following condition  from \cite{Makinson03} on any spd-algebra $\mathbb{H}$ based on a De Morgan\footnote{a {\em De Morgan} algebra ia a distributive lattice endowed with a unary operation which is involutive and both finitely meet- and join-reversing.}  algebra:

\[\mathrm{(CTA)}^{-1}\quad\quad a \prec x\ \&\ a \pcon  \neg( x \wedge y)\Rightarrow  (a \wedge x) \pcon \neg y.\]

{{
\centering
\begin{tabular}{rl l}
&  $\forall x \forall y \forall a (a \prec x\ \&\ a \pcon  \neg( x \wedge y)\Rightarrow  (a \wedge x) \pcon \neg y)$ \\
iff &  $\forall x \forall y \forall a ( a \prec x\ \&\ (a \wedge x) \npcon \neg y \Rightarrow  a \npcon  \neg( x \wedge y))$ \\
iff &  $\forall x \forall y \forall a ( a \prec x\ \&\ (a \wedge x) \npcon \neg y \Rightarrow  a \npcon  \neg x \vee \neg y)$ \\
iff &  $\forall x \forall y \forall a ( a \prec x\ \&\ (a \wedge x) \npcon \neg y \Rightarrow  a \npcon  \neg x \ \&\  a\npcon \neg y)$, & (C2)\\

\end{tabular}
\par
}}
\smallskip

\noindent which is readily equivalent to the conjunction of the two conditions $a \prec x \ \& \ (a \wedge x) \npcon \neg y \Rightarrow a \npcon \neg x$ and $a \prec x \ \& \ (a \wedge x) \npcon \neg y \Rightarrow a \npcon \neg y$, on which we proceed separately. The first quasi-inequality satisfies Definition \ref{def:inverse_shape} when $a$ is regarded as an $\overline a$-variable, $y$ as a $\overline c$-variable, and $x$ as a $\overline v$-variable. Applying Algorithm \ref{algo:inverse} to it yields:

{{\centering
\begin{tabular}{rll}
& $\forall x \forall y \forall a (a \prec x \ \& \ (a \wedge x) \npcon \neg y \Rightarrow a \npcon \neg x)$ \\
iff & $\forall x \forall y \forall a (a \leq \blacksquare x\ \&\ \neg y \leq {\rhd} (a \wedge x)\Rightarrow  a \leq {\blacktriangleright}  \neg x)$ \\
iff & $\forall x \forall y \forall a (a \leq \blacksquare x\ \&\ \neg {\rhd} (a \wedge x) \leq y \Rightarrow  a \leq {\blacktriangleright}  \neg x)$ \\
iff & $\forall x  \forall a (a \leq \blacksquare x\ \&\ \exists y(\neg {\rhd} (a \wedge x) \leq y) \Rightarrow  a \leq {\blacktriangleright}  \neg x)$ \\
iff & $\forall x  \forall a (a \leq \blacksquare  x \Rightarrow a \leq {\blacktriangleright} \neg x)$ & ($\ast$) \\
iff & $\forall x  (\blacksquare x \leq {\blacktriangleright}\neg x)$ &
\end{tabular}
\par}}
\smallskip

\noindent The equivalence marked with ($\ast$) follows immediately from the fact that  $\exists y(\neg {\rhd} (a \wedge x) \leq y)$ is always satisfied, since $A$ is a bounded lattice.
The second quasi-inequality satisfies Definition \ref{def:inverse_shape} when $x$ is regarded as a $\overline c$-variable, $y$ as a $\overline b$-variable, and $a$ as a $\overline v$-variable. Applying Algorithm \ref{algo:inverse} to it yields:

\medskip

{{\centering
\begin{tabular}{rl}
& $\forall x \forall y \forall a (a \prec x \ \& \ (a \wedge x) \npcon \neg y \Rightarrow a \npcon \neg y)$ \\
iff & $\forall x \forall y \forall a (a \leq \blacksquare x\ \&\ \neg y \leq {\rhd} (a \wedge x)\Rightarrow  a \leq {\blacktriangleright}  \neg y)$ \\
iff &  $\forall x \forall y \forall a (\Diamond a \leq x\ \&\ \neg {\rhd} (a \wedge x) \leq y \Rightarrow  a \leq {\blacktriangleright} \neg y)$ \\
iff &  $\forall x \forall y \forall a (\neg {\rhd} (a \wedge \Diamond a) \leq y \Rightarrow  a \leq {\blacktriangleright} \neg y)$ \\
iff & $\forall x \forall y \forall a (a \leq {\blacktriangleright} \neg \neg {\rhd} (a \wedge \Diamond a))$\\
iff & $\forall a ({\rhd} (a \wedge \Diamond a) \leq {\rhd} a)$
\end{tabular}
\par}}
\end{example}

\begin{example}
The formula $\forall v\forall a \forall c_1\forall c_2 (a \leq c_1 \rightarrow c_2 \ \& \  v \prec c_1 \ \& \ c_2 \prec v \Rightarrow a \prec v)$ is a Kracht formula when equivalently rewritten as $\forall v\forall a (\forall c_1, c_2 \geq_\rightarrow a)(v \prec c_1 \ \& \ c_2 \prec v \Rightarrow a \prec v)$. Applying Algorithm \ref{algo:inverse} to it yields:

{{\centering
\begin{tabular}{rl}
& $\forall v\forall a \forall c_1\forall c_2 (a \leq c_1 \rightarrow c_2 \ \& \  v \prec c_1 \ \& \ c_2 \prec v \Rightarrow a \prec v)$ \\
iff & $\forall v\forall a \forall c_1 \forall c_2 (a \leq c_1 \rightarrow c_2 \ \& \  \Diamond v \leq c_1 \ \& \ c_2 \leq \blacksquare v \Rightarrow a \leq \blacksquare v)$ \\
iff & $\forall v\forall a (a \leq \Diamond v \rightarrow \blacksquare v \Rightarrow a \leq \blacksquare v)$ \\
iff & $\forall v(\Diamond v \rightarrow \blacksquare v \leq \blacksquare v)$. \\
\end{tabular}
\par}}
\end{example}

\begin{example}
\label{ex:big_arity_uniform} Let us illustrate  how variables in $\overline d$ (resp.~$\overline c$) can be eliminated when the relation in their restricting inequality is $\leq_f$ or $\geq_g$ with $n_f, n_g\geq 2$,  and there is at least one variable occurring with the same polarity  in the restrictor and in $\zeta$ (resp.~$\eta$). The  formula $\forall v \exists d_1 \exists d_2(d_1 {\small\pdra} d_2 \leq v \ \& \ d_1 \leq v \ \& \ d_2 \leq v \ \& \ v\prec v)$, where $\pdra\in \mathcal{F}$ and s.t.~$\varepsilon_{{\tiny\pdra}} = (\partial, 1)$, is a Kracht formula when equivalently rewritten as $\forall v(\exists (d_1, d_2) \leq_{\tiny\pdra} v)(d_1 \leq v \ \& \ d_2 \leq v \ \& \ v\prec v)$.
\smallskip

{{\centering
\begin{tabular}{rll}
 & $\forall v \exists d_1 \exists d_2(d_2 \leq d_1 \vee v \ \& \ d_1 \leq v \ \& \ d_2 \leq v \ \& \ v\prec v)$ \\
iff & $\forall v \exists d_1 \exists d_2(d_2 \leq (d_1 \vee v) \wedge v \ \& \ d_1 \leq v \ \& \ \Diamond v\leq v)$ \\
iff & $\forall v \exists d_1 (\exists d_2(d_2 \leq (d_1 \vee v) \wedge v) \ \& \ d_1 \leq v \ \& \ \Diamond v\leq v)$ \\
iff & $\forall v\exists d_1 (d_1 \leq v \ \& \ \Diamond v\leq v)$ & ($\ast$) \\
iff & $\forall v ( \Diamond v\leq v).$ & ($\ast$) \\
\end{tabular}
\par}}
\smallskip

The equivalences marked with ($\ast$) follow immediately from the fact that $\exists d_2(d_2 \leq (d_1 \vee v) \wedge v)$ (resp.~$\exists d_1 (d_1 \leq v)$) is always satisfied in a bounded lattice by instantiating $d_2$ (resp.~$d_1$) to $\bot$.
\end{example}

\begin{example}
    \label{ex: elimination d main case}
    The formula $\forall v_1 \forall v_2 \forall b(v_1\wedge v_2\prec b \Rightarrow \exists d_1 \exists d_2(  v_1\prec d_1\ \&\ v_2\prec d_2 \ \&\ d_1\wedge d_2 \prec b))  $ is a Kracht formula when equivalently rewritten as $\forall v_1 \forall v_2 \forall b(v_1\wedge v_2\prec b \Rightarrow (\exists d_1\succ v_1) (\exists d_2\succ v_2)(d_1\wedge d_2 \prec b))$.
    Let us illustrate how to eliminate the variables in $\overline d$ when all their occurrences in $\zeta$ have opposite polarity to their occurrences in the restrictors.
    \smallskip

    {{\centering
\begin{tabular}{rll}
& $\forall v_1 \forall v_2 \forall b(v_1\wedge v_2\prec b \Rightarrow \exists d_1 \exists d_2(  v_1\prec d_1\ \&\ v_2\prec d_2 \ \&\ d_1\wedge d_2 \prec b))  $\\
iff & $\forall v_1 \forall v_2 \forall b(\Diamond(v_1\wedge v_2)\leq b \Rightarrow \exists d_1 \exists d_2( \Diamond v_1\leq d_1\ \&\ \Diamond v_2\leq d_2 \ \&\ \Diamond (d_1\wedge d_2) \prec b))  $\\
iff & $\forall v_1 \forall v_2 \forall b(\Diamond(v_1\wedge v_2)\leq b \Rightarrow \exists d_1 ( \Diamond v_1\leq d_1\ \&\ \exists d_2(\Diamond v_2\leq d_2 \ \&\ \Diamond (d_1\wedge d_2) \leq b)))  $\\
iff & $\forall v_1 \forall v_2 \forall b(\Diamond(v_1\wedge v_2)\leq b \Rightarrow \exists d_1 ( \Diamond v_1\leq d_1\ \&\  \Diamond (d_1\wedge \Diamond v_2) \leq b))  $ & Lemma \ref{lem: diamond-output equivalence extended}\\
iff & $\forall v_1 \forall v_2 \forall b(\Diamond(v_1\wedge v_2)\leq b \Rightarrow  \Diamond (\Diamond v_1\wedge \Diamond v_2) \leq b))  $ & Lemma \ref{lem: diamond-output equivalence extended}\\
iff & $\forall v_1 \forall v_2 (\Diamond (\Diamond v_1\wedge \Diamond v_2) \leq\Diamond(v_1\wedge v_2)).  $\\
 \end{tabular}
\par}}
\end{example}

\begin{example}
\label{ex:first approx}
The formula $\forall a \forall b \forall v(a \prec v\ \& \   v\prec b\Rightarrow \exists d( b\nppcon d\ \&\ d\npcon  a))$ is a Kracht formula. Applying Algorithm \ref{algo:inverse} to it yields:

       {{\centering
\begin{tabular}{rll}
 & $\forall a \forall b \forall v(a \prec v\ \& \   v\prec b\Rightarrow \exists d( b\nppcon d\ \&\ d\npcon  a))$\\
iff & $\forall a \forall b \forall v(a\leq \blacksquare v\ \&\ \Diamond v \leq b\Rightarrow \exists d(\tb b\leq d\ \&\ d\leq \bt a))$\\
iff & $\forall a \forall b \forall v(a\leq \blacksquare v\ \&\ \Diamond v \leq b\Rightarrow \tb b\leq \bt a)$\\
iff & $\forall v(\tb \Diamond v\leq \bt \blacksquare  v).$ & Lemma \ref{lemma:first_approx}\\
 \end{tabular}
\par}}
\end{example}

\begin{example}
The formula
$\forall x\exists y(x \nppcon y \ \&\ y \npcon x)$ is  Kracht when regarding $y$ as a $\overline{d}$-variable.  Applying Algorithm \ref{algo:inverse} to it yields:

\medskip

{{\centering
\begin{tabular}{rl}
& $\forall x\exists y(x \nppcon y \ \&\ y \npcon x)$ \\
iff & $\forall x\exists y( \tb x \leq y \ \&\  y \leq \bt x)$\\
iff & $\forall x( \tb x \leq \bt x)$.\\
\end{tabular}
\par}}
The formula
$\forall x\forall y(x \nppcon y \Rightarrow y \npcon x)$ is not Kracht; indeed, neither variable can be regarded as an $\overline{a}$- or a $\overline{b}$-variable, since they both violate condition 4 on the polarity
(this is better seen after the translation step, which yields
$\forall x\forall y( \tw y \leq x \Rightarrow x \leq \wt y)$). However, this violates condition 8 of Definition \ref{def:inverse_shape}, which prevents  Algorithm \ref{algo:inverse}  to be applied successfully.
\end{example}

\section{Inverse correspondence procedure}
\label{sec:inverse_correspondence_procedure}

In this section, we describe and prove the soundness of an algorithm which takes Kracht formulas (cf.~Definition \ref{def:inverse_shape}) in input, and  outputs a (conjunction of) $\mathcal{L}_\mathrm{LE}$-inequalities of which the given Kracht formula is the first order correspondent.

\paragraph{Preprocessing} Firstly, every relational atom $sRt$ in $\eta$ and $\zeta$ is translated into some $\mathcal{L}_\mathrm{LE}$-inequality (cf.\ table above Definition \ref{def: sigma and pi extensions of slanted}) using the map $\rho$ described in Section \ref{sec:kracht_formulas}. Note that the procedure described in this section is indifferent to how atoms are translated.  

For every atom $sRt$ in $\zeta$ containing an occurrence of a variable in $\overline d$, apply residuation and splitting as needed in $\rho(sRt)$ so as to display $d$ as indicated in the lemma below.

\begin{lemma}
\label{lemma:atoms_in_zeta_preprocessed}
If  $sRt$ is an atom in $\zeta$ containing a positive (resp.\ negative) occurrence of a variable $d$ in $\overline d$, then   $\rho(sRt)$  is equivalent to a conjunction of clopen-Skeleton inequalities $(\bigwith_i d \leq \theta_i)\ \&\ (\bigwith_j \theta'_j\leq \theta_j)$ (resp.\ $(\bigwith_i \theta'_i \leq d) \ \&\ (\bigwith_j \theta'_j\leq \theta_j)$) such that every $\theta$ (resp.~$\theta'$) 
does not contain any variable in $\overline d$, and contains at most one  occurrence of each variable in $\overline a$ and $\overline b$.
\end{lemma}
\begin{proof}
By Definition \ref{def:inverse_shape}.9, each atom $sRt$ in $\zeta$ contains at most one occurrence of a variable in $\overline{d}$, and the first common ancestor in the signed generation tree of $sRt$ of any two different occurrences of the same variable  in $\overline a$ or $\overline b$ is either $+\wedge$ or $-\vee$. If $d$ occurs in $s$, then we proceed by induction on $s$. If $s\coloneqq d$, then we can equivalently transform $\rho (sRt)$ into $d\leq t$, or $t\leq d$, or,  by applying residuation if needed, into either $d\leq g(t)$ for some $g\in \mathcal{G}_{\mathsf{S}}\cup\mathcal{G}_{\mathsf{C}}$, or $f(t)\leq d$ for some $f\in \mathcal{F}_{\mathsf{S}}\cup\mathcal{F}_{\mathsf{D}}$. In either case, the ensuing inequalities immediately satisfy the conditions of statement. 
If $s\coloneqq s_1(!d)\wedge s_2$, then, since $d$ is displayable by Definition \ref{def:inverse_shape}.7,  modulo suitable application of residuation (cf.~Footnote \ref{ftn:display or residuation}) if needed, $\rho (sRt)$ can be equivalently transformed into either $t\leq s_1(!d)\wedge s_2$, and hence into the conjunction of $t\leq s_1(!d)$ and $t\leq s_2$,  or into $f(t)\leq s_1(!d)\wedge s_2$ for some $f\in \mathcal{F}_{\mathsf{D}}$, and hence into the conjunction of $f(t)\leq s_1(!d)$ and $f(t)\leq s_2$. In either case, all the ensuing inequalities are clopen-Skeleton and satisfy the conditions on the variables in $\overline d$, $\overline a$ and $\overline b$ stated in the lemma; hence, by applying the induction hypothesis to the inequality in which  $s_1(!d)$  occurs, we obtain the required statement. In the distributive setting, modulo suitable application of residuation if needed, $\rho (sRt)$ can also be equivalently transformed into $s_1(!d)\wedge s_2\leq t$, or into $s_1(!d)\wedge s_2\leq g(t)$ for some $g\in \mathcal{G}_{\mathsf{S}}$. Hence, the statement follows by applying induction hypothesis to $s_1(!d)\leq \ s_2\to t$ or to $s_1(!d)\leq \ s_2\to g(t)$. The remaining cases are similar.
\end{proof}

Similarly, for each atom in $\eta$, the displayable occurrence of a variable in $\overline a$, $\overline b$, or $\overline c$ can be displayed, as indicated in the next lemma.
\begin{lemma}
\label{lemma:atoms_in_eta_preprocessed}
Let $sRt$ be an atom in $\eta$ containing an occurrence of a variable $u$ in $\overline a$ (resp.\ $\overline b$) or a positive (resp.\ negative) occurrence of variable in $\overline c$. Then,  $\rho(sRt)$ is equivalent to a conjunction of clopen-Skeleton inequalities $(\bigwith_i u \leq \theta_i)\ \&\ (\bigwith_j \theta'_j\leq \theta_j)$ (resp.\ $(\bigwith_i \theta'_i \leq u) \ \&\ (\bigwith_j \theta'_j\leq \theta_j)$) such that every $\theta$ (resp.~$\theta'$) 
only contains variables in $\overline v$. 
\end{lemma}
\begin{proof}
By  Definition \ref{def:inverse_shape}.8, each atom $sRt$ in $\eta$ contains exactly one occurrence of a variable $u$ in $\overline{a}$, $\overline{b}$ or $\overline{c}$. If $u$ occurs in $s$, then we proceed by induction on $s$. If $s\coloneqq u$, then, by applying residuation if needed, we can equivalently transform $\rho (sRt)$ into either $u\leq g(t)$ for some $g\in \mathcal{G}_{\mathsf{S}}\cup\mathcal{G}_{\mathsf{C}}$, or $f(t)\leq u$ for some $f\in \mathcal{F}_{\mathsf{S}}\cup\mathcal{F}_{\mathsf{D}}$, or $u\leq t$, or $t\leq u$. In either case, the resulting inequalities are clopen-Skeleton, and by  Definition \ref{def:inverse_shape}.8, only contain variables in $\overline{v}$. As to the inductive step, we proceed by cases. If $s\coloneqq s_1(!u)\vee s_2$, then, since $u$ is displayable by  Definition \ref{def:inverse_shape}.7,  modulo suitable application of residuation if needed, $\rho (sRt)$ can  equivalently be transformed into either $s_1(!u)\vee s_2\leq t$, and hence into the conjunction of $ s_1(!u)\leq t$ and $ s_2\leq t$,  or into $s_1(!u)\vee s_2\leq g(t)$ for some $g\in \mathcal{G}_{\mathsf{S}}$, and hence into the conjunction of $ s_1(!u)\leq g(t)$ and $s_2\leq g(t)$. In either case, all the ensuing inequalities are clopen-Skeleton and contain only variables in $\overline{v}$ besides the single occurrence of $u$; then, by applying the induction hypothesis to the inequality in which  $s_1(!u)$  occurs we obtain the required statement. In the distributive setting, modulo suitable application of residuation if needed, $\rho (sRt)$ can also equivalently be transformed into $t\leq s_1(!u)\vee s_2$, or into $f(t)\leq s_1(!u)\vee s_2$ for some $f\in \mathcal{F}_{\mathsf{D}}$. Hence, the statement follows by applying the induction hypothesis to $t\pdla s_2\leq s_1(!u)$ or to $f(t)\pdla s_2\leq s_1(!u)$. The remaining cases are similar.
\end{proof}

After displaying all occurrences of variables in $\overline c$ (resp.\ $\overline d$)  in the $\rho$-translation of atoms in $\eta$ (resp.\ $\zeta$) as indicated in the lemmas above, 
we equivalently replace the conjunctions $\bigwith_j c \leq \xi_j$ or $\bigwith_j \xi_j \leq c$ (resp.\ $\bigwith_j d \leq \xi_j$ or $\bigwith_j \xi_j \leq d$)  with  single inequalities $c \leq \bigwedge_j \xi_j$ or $\bigvee_j \xi_j \leq c$ (resp.\ $d \leq \bigwedge_j \xi_j$ or $\bigvee_j \xi_j \leq d$ ). 
\medskip


At the end of the preprocessing step, the translated $\eta$ (resp.\ $\zeta$) is rewritten as a conjunction of inequalities $\eta^p$ (resp.\ $\zeta^p$) in which every variable in $\overline c$ (resp. $\overline d$) occurs at most once, and in display.

\paragraph{Eliminating variables in $\overline d$}
After the pre-processing step, all the variables in $\overline d$ and their (restricted) existential quantification are eliminated, as indicated in the following proposition.

\begin{proposition}
\label{prop:flattify_consequent}
The consequent $(\exists \overline{d_o} R_o y_o)\cdots(\exists \overline{d_1} R_1 y_1) \zeta^p$ is equivalent to a conjunction of inequalities $\zeta^f$ each of which contains only variables from $\overline a$, $\overline b$, and $\overline v$, and is such that  any two  occurrences of the same variable in $\overline a$ or $\overline b$ have a $-\wedge$ or $+\vee$ node as first common ancestor.
\end{proposition}
\begin{proof}
We show by induction that, for every $0\leq i\leq o$, $(\exists \overline{d_i} R_i y_i)\cdots(\exists \overline{d_1} R_1 y_1) \zeta^p$ is equivalent to a conjunction of inequalities where no variable in any $\overline{d_j}$ occurs for $j \leq i$, and where every variable in any $\overline{d_h}$ with $h > i$ occurs at most once and displayed.

The base case (when $i = 0$) is obvious since the preprocessing guarantees that every variable in every $\overline{d_h}$ with $h > i$ occurs at most once and displayed in  $\zeta^p$.

When $i \geq 1$, by inductive hypothesis $(\exists \overline{d_i} R_i y_i)\cdots(\exists \overline{d_1} R_1 y_1) \zeta^p$ can be rewritten as $(\exists \overline{d_i} R_i y_i)\zeta'$, where $\zeta'$ contains no occurrence of variables in any $\overline{d_j}$ with $j \leq i$ and where all variables in any $\overline{d_h}$ with $h > i$ occurs at most once and displayed. We proceed by cases: If no variable in $\overline{d_i}$ occurs in $\zeta'$, then $(\exists \overline{d_i} R_i y_i)\zeta'$ is equivalent to $\zeta'$, which satisfies the statement.
Let us assume that at least some variable in $\overline{d_i}$ occurs in $\zeta'$; then all the variables in $\overline{d_i}$ occurring also in $\zeta'$ occur in some inequality $\overline{\xi(d_i, \theta)}$, and there is at least one inequality in such a vector $\overline{\xi}$. Then, $(\exists \overline{d_i} R_i y_i)\cdots(\exists \overline{d_1} R_1 y_1) \zeta^p$ can be equivalently rewritten as
\[
\exists \overline{d_i}( \rho(\overline{d_i}R_iy_i) \ \& \ \overline{\xi(d_i, \theta)}) \ \& \ \zeta'^{-} ),
\]
where $\zeta'^{-}$ is $\zeta'$ without the inequalities in $\overline{\xi(d_i, \theta)}$. There are two subcases: if the polarity of some variable $d$ in $\overline{d_i}$ occurring in $\rho(\overline{d_i}R_iy_i)$ is the same as its polarity in $\xi(d_i, \theta)$, then the conjuncts $\rho(\overline{d_i}R_iy_i) $ and $ \xi(d, \theta)$ can be rewritten as a single inequality $\theta \vee h \leq d$ or $d \leq \theta \wedge h$, where $h$ is the term obtained by displaying $d$ in $\rho(\overline{d_i}R_iy_i)$. Then, $\exists d(\theta \vee h \leq d)$ (resp.\ $\exists d(d \leq \theta \wedge h)$) can be removed, since it is always satisfied in a bounded lattice by instantiating $d$ to $\top$ (resp.~$\bot$).
Once $d$ has been eliminated, by Definition \ref{def:inverse_shape}.5, the remaining variables in $\overline{d_i}$ occur unrestricted, in display, and  with uniform polarity. Therefore, their corresponding existential quantification can be moved under the scope of the conjunction of atoms in $\zeta'$, and each formula $\exists d_i\xi(d_i, \theta)$ is always  satisfied by instantiating $d_i$ to $\top$ or $\bot$, depending on its polarity (cf.~Example \ref{ex:big_arity_uniform} for an illustration of this situation).
Hence, in this case, the statement is proved with $\zeta'^{-}$ as the witness.

If the polarity of all variables in $d_i$ in $\rho(\overline{d_i}R_iy_i)$ is opposite to their polarity in the corresponding inequalities in $\overline{\xi(d_i, \theta)}$,
by Lemma \ref{lem: diamond-output equivalence extended}, $\exists \overline{d_i}( \rho(\overline{d_i}R_iy_i) \ \& \ \overline{\xi(d_i, \theta)})$ can equivalently be rewritten as indicated in the following table.

\begin{table}[h]
\centering
\begin{tabular}{rcl}
$\exists d_i(y_i \leq \blacksquare d_i \ \& \ d_i \leq \theta)$ & iff & $y_i \leq \blacksquare \theta$  \\
$\exists d_i(y_i \leq {\rhd} d_i \ \& \ d_i \leq \theta)$ & iff & $y_i \leq {\rhd} \theta$  \\
$\exists d_i(y_i \leq {\blacktriangleright} d_i \ \& \ d_i \leq \theta)$ & iff & $y_i \leq {\blacktriangleright} \theta$  \\
$\exists d_i(y_i \leq  d_i \ \& \ d_i \leq \theta)$ & iff & $y_i \leq  \theta$  \\
$\exists \overline{d_i}(y_i \leq g(\overline{d_i}) \ \& \ \overline{d_i \leq^{\varepsilon_g} \theta})$ & iff & $y_i \leq g(\overline \theta)$ \\
\end{tabular}
\quad\quad
\begin{tabular}{rcl}
$\exists d_i(\Diamond d_i \leq y_i \ \& \ \theta \leq d_i)$ & iff & $\Diamond \theta \leq y_i$  \\
$\exists d_i({\lhd} d_i \leq y_i \ \& \ \theta \leq d_i)$ & iff & ${\lhd} \theta \leq y_i$  \\
$\exists d_i({\blacktriangleleft} d_i \leq y_i \ \& \ \theta \leq d_i)$ & iff & ${\blacktriangleleft} \theta \leq y_i$  \\
$\exists d_i(d_i \leq y_i \ \& \ \theta \leq d_i)$ & iff & $\theta \leq y_i$  \\
$\exists \overline{d_i}(f(\overline{d_i}) \leq y_i \ \& \ \overline{\theta \leq^{\varepsilon_f} d_i})$ & iff & $f(\overline \theta)\leq y_i$ \\
\end{tabular}
\caption{Equivalences that allow to eliminate the quantifiers.}
\label{table:deflattification}
\end{table}

\noindent In each case, the equivalent rewriting results in an inequality $\chi(y_i)$ with $y_i$ in display. If $y_i$ is not a variable in $\overline d$, then the conjunction $\chi(y_i) \ \& \ \zeta'^-$ is the required witness. If $y_i$ is a variable in $\overline d$, i.e., it is a $d_h$ with $h > i$, it remains to show that $\chi(y_i) \ \& \ \zeta'^-$ can be rewritten as a conjunction of inequalities where $d_h$ occurs at most once and is displayed (cf.~Example \ref{ex: elimination d main case} for an illustration of this situation). Since $d_h$ occurs in display  in some inequality $\xi'(d_h, \theta')$, $\zeta'^-$ can be rewritten as $\zeta'' \ \& \ \xi'(d_h, \theta')$. Therefore, $\chi(y_i) \ \& \ \zeta'^-$ can be rewritten as $\chi(d_h) \ \& \ \zeta'' \ \& \ \xi'(d_h, \theta')$.  Definition \ref{def:inverse_shape}.6 guarantees that the polarity of $d_h$ in $\chi(d_h)$ and $\xi'(d_h, \theta')$ is the same; thus, their conjunction can be rewritten as an inequality $\theta' \vee \theta'' \leq d_h$ or $d_h \leq \theta' \wedge \theta''$, which we denote  $\xi''(d_h)$. The conjunction $\xi''(d_h) \ \& \ \zeta''$  is the witness that proves the statement in this case.

\medskip

The condition that any two  occurrences of the same variable in $\overline a$ or $\overline b$ have a $-\wedge$ or $+\vee$ node as first common ancestor in the resulting conjunction of inequalities follows straightforwardly from the fact that the starting $\zeta^p$ has this property,  the fact that the polarities of the variable occurrences are invariant under the applications of equivalent rewritings, and from the fact that terms containing variables in $\overline a$ and $\overline b$ are only substituted in flat terms of shape $h(y_i)$; therefore, there are no more pairs of occurrences of variables in $\overline a$ and $\overline b$ for which the property has to be checked.
\end{proof}

\begin{remark}
At the end of the process described in the proposition above, each conjunct $\zeta_i^f$ in $\zeta^f$ is an inequality containing only variables from $\overline a$, $\overline b$, and $\overline v$. We remind the reader that, by Definition \ref{def:inverse_shape}.7, each occurrence of a variable in $\overline a$ and $\overline b$ is displayable, i.e., all its ancestors are skeleton nodes in the signed generation tree of the inequality. This property is preserved by Proposition \ref{prop:flattify_consequent} because every applied rewriting preserves the polarity of existing nodes, and only introduces skeleton nodes by substituting positive (resp.\ negative) skeletons in positive (resp.\ negative) positions, or by means of `inverse splitting rules'.\footnote{\label{footnote:inverse_splitting} These rules allow us to equivalently rewrite $a\leq \gamma_1\ \&\ a\leq \gamma_2$ as $a\leq \gamma_1\wedge\gamma_2$ and $\delta_1\leq b\ \&\  \delta_2\leq b$ as $ \delta_1\vee\delta_2\leq b$.}
\end{remark}

\paragraph{Eliminating variables in $\overline c$}

Similarly to the step above, all the variables  in $\overline c$ and their restricted quantification are eliminated, as indicated in the following proposition.
\begin{proposition}
\label{prop:flattify_antecedent}
The formula $(\exists \overline{c_m}R'_m z_m)\cdots(\exists \overline{c_1} R'_1 z_1) \eta^p$
is equivalent to a conjunction $\eta^f$ of inequalities of shape $a \leq \theta$ or $\theta \leq b$, where $a$ is a variable in $\overline a$, $b$ is a variable in $\overline b$, and $\theta$ is a term containing only variables in $\overline v$.
\end{proposition}
\begin{proof}
The proof of the variable-elimination is analogous to that of Proposition \ref{prop:flattify_consequent}. It remains to be shown that every conjunct in $\eta^f$ contains exactly one variable occurrence in $\overline a$ or $\overline b$.
To show that there is at most one  such occurrence, it is sufficient to notice that every term substituted into a restricting inequality only contains variables from $\overline v$, by Definition \ref{def:inverse_shape}.8.
To show that there is at least one  such occurrence, it is sufficient to note that whenever a term is substituted as in Table \ref{table:deflattification} (with $c_i$ replacing $d_i$ and $z_i$ replacing $y_i$), a single inequality is obtained where some $z_i$ occurs displayed. Hence, at the end of this process, in the only remaining inequalities $z_i$ will be a variable in $\overline a$ or $\overline b$ by Definition \ref{def:inverse_shape}.2.

This proves the statement for all the inequalities affected by the process in Proposition \ref{prop:flattify_consequent}. Let us show the statement  for the remaining inequalities, i.e., those which do not contain any occurrence of variables in $\overline c$ in $\eta^p$, but contain exactly one occurrence of a variable in either $\overline a$ or $\overline b$ (by Definition \ref{def:inverse_shape}.8). By Definition \ref{def:inverse_shape}.7, for each  such inequality, this occurrence of a variable in $\overline a$ or $\overline b$ is displayable; hence, the inequality can  equivalently be rewritten in the required shape (see Lemma \ref{lemma:displayable}).
\end{proof}

\paragraph{Eliminating variables in $\overline a$ and $\overline b$} The process above equivalently rewrites the input inequality as follows:
\begin{equation}
\label{eq:f-formula}
    \forall \overline v \forall \overline a \forall \overline b \left( \eta^f \Rightarrow \zeta^f \right),
\end{equation}
where $\zeta^f$ has the properties stated in Proposition \ref{prop:flattify_consequent}, and $\eta^f$ has the properties stated in Proposition \ref{prop:flattify_antecedent}.

\begin{lemma}
\label{lemma:rewrite_as_conjunction}
The formula \eqref{eq:f-formula} above can  equivalently be rewritten as a conjunction of quasi-inequalities
\begin{equation}
\label{eq:s-formula}
    \bigwith_i \forall \overline v\forall \overline a\forall \overline b \left(\eta^{s} \Rightarrow \zeta^{s}_i \right),
\end{equation}
such that $\eta^s$ is a conjunction of clopen-Skeleton inequalities, each of which contains exactly one displayed occurrence of a variable in $\overline a$ or $\overline b$, and each $\zeta_i^s = \zeta_i^s(\overline a, \overline b, \overline v)$ is a skeleton inequality in which every  variable in $\overline a$ and $\overline b$ occurs at most once.
\end{lemma}
\begin{proof}
By Definition \ref{def:inverse_shape}.4, all occurrences of variables in $\overline a$ (resp.~$\overline b$) are positive (resp.~negative); hence, by applying the inverse splitting rules (cf.\ Footnote \ref{footnote:inverse_splitting}), $\eta^f$ can  equivalently be rewritten  so as to merge all inequalities sharing the same variable in $\overline a$ or $\overline b$ into one single inequality. More specifically, for each $a$ in $\overline a$ (resp.\ $b$ in $\overline b$) occurring in inequalities $a \leq \theta_i$ (resp.\ $\theta_i \leq b$) in $\eta^f$, the conjunction $\bigwith_i a \leq \theta_i$ (resp.\ $\bigwith_i \theta_i \leq b$) is rewritten as $a \leq \bigwedge_i \theta_i$ (resp.\ $\bigvee_i \theta_i \leq b$). Let us denote by $\eta^{s}$ the result of this process.

Each conjunct $\zeta^f_i$ in  $\zeta^f$ can  equivalently be rewritten by applying the distribution laws of all connectives so as to push $-\wedge$ and $+\vee$ nodes above occurrences of variables in $\overline a$ and $\overline b$. Since every two occurrences of the same variable in $\overline a$ and $\overline b$ share a node $-\wedge$ or $+\vee$ as first common ancestor (cf.~Proposition \ref{prop:flattify_consequent}), by exhaustively applying the splitting rules,  the  resulting inequality can  equivalently be rewritten as  a conjunction of inequalities, in each of which, each variable in $\overline a$ and $\overline b$ occurs at most once. Let $\zeta^{s}$  denote  the result of applying this process to each conjunct in $\zeta^f$, and let $\zeta^{s}_i$ denote   the conjuncts of $\zeta^{s}$.
It follows that \eqref{eq:f-formula} is equivalent to the conjunction of quasi-inequalities
$\bigwith_i \forall \overline v\forall \overline a\forall \overline b \left(\eta^{s} \Rightarrow \zeta^{s}_i \right)$.
\end{proof}

\begin{remark}
\label{remark:as_and_bs_occur_everywhere_wlog}
Thanks to the proposition above, we can proceed on each  quasi-inequality in \eqref{eq:s-formula} separately. Notice that, when working on each such quasi-inequality, the variable occurrences in it might take a different role. This is the case  of the  variables in $\overline a$ and $\overline b$, which can occur either in both $\eta^s$ and $\zeta^s_i$, or only in one of them (or in neither of them, in which case we omit the quantification). In each quasi-inequality,
each variable in $\overline a$ or in $\overline b$ which occurs only in 
$\zeta^s_i$ is then reassigned to the set of   variables in $\overline v$. Notice that variables in $\overline v$ do not have to satisfy any constraints, hence this reassignment is harmless; moreover, since by assumption the given variable does not occur in $\eta^s$, after the reassignment, it is still true that each inequality in $\eta^s$ contains exactly one displayed occurrence of a variable in $\overline a $ or $\overline b$, that is, the quasi inequality after the reassignment still satisfies  the statement of Lemma \ref{lemma:rewrite_as_conjunction}. Thus, modulo this reassignment, we can assume without loss of generality that every variable in $\overline a$ and $\overline b$ occurs in $\eta^s$.

Furthermore, we can assume without loss of generality that every variable in $\overline a$ and $\overline b$ occurs in $\zeta^s_i$, since, for any $a$ in $\overline a$ (resp.\ $b$ in $\overline b$) which does not occur in $\zeta^s_i$, $\bigwith_i \forall \overline v\forall \overline a\forall \overline b \left(\eta^{s} \Rightarrow \zeta^{s}_i \right)$ is equivalent to $\bigwith_i \forall \overline v\forall (\overline a)^-\forall \overline b \left(\exists a(\eta^{s}) \Rightarrow \zeta^{s}_i \right)$ (resp.\ $\bigwith_i \forall \overline v\forall\overline a\forall (\overline b)^- \left(\exists b(\eta^{s}) \Rightarrow \zeta^{s}_i \right)$), where $(\overline a)^-$ is $\overline a$ without $a$ (resp.\ $(\overline b)^-$ is $\overline b$ without $b$).
Then, $\exists a(\eta^{s})$ (resp.~$\exists b(\eta^{s})$) is equivalent to $\eta^{s-}$, which is the conjunction of inequalities $\eta^s$ minus the inequality in which $a$ (resp.~$b$) occurs, since, similarly to what was discussed in the proof of Proposition  \ref{prop:flattify_consequent}, $\exists a (a\leq t)$ (resp.~$\exists b(t\leq b)$) holds  true under any valuation $v: \mathsf{Prop}\to A$ (cf.~Proposition \ref{prop: positive PIA are open}.1).
\end{remark}

\begin{proposition}
\label{prop:inverse_first_approx}
Every quasi-inequality $\forall \overline v\forall \overline a\forall \overline b \left(\eta^{s} \Rightarrow \zeta^{s}_i \right)$ is equivalent to a clopen-analytic $\mathcal{L}_\mathrm{LE}$-inequality.
\end{proposition}
\begin{proof}
By Remark \ref{remark:as_and_bs_occur_everywhere_wlog}, we may assume w.l.o.g.\ that all variables in $\overline a$ and $\overline b$ occur exactly once in $\eta^s$ and exactly once in $\zeta^s_i$. Hence, the quasi-inequality can be rewritten as follows
\[
\forall \overline v\forall \overline a\forall \overline b \left(
\overline{a \leq \alpha}
\ \& \
\overline{\beta \leq b}
\Rightarrow \zeta^{s}_i(!\overline a, !\overline b, \overline v)
\right).
\]
Since $\zeta_i^s$ is a Skeleton inequality, by Lemma \ref{lemma:first_approx}, the quasi inequality is equivalent to
\[
\forall \overline v \quad
\zeta^{s}_i[\overline \alpha/!\overline a, \overline \beta/!\overline b],
\]
which is clopen-analytic since each $\alpha$ in $\overline \alpha$ (resp.\ $\beta$ in $\overline \beta$) is a clopen negative (resp.\ positive) skeleton formula containing only variables in $\overline v$ (cf.~Example \ref{ex:first approx} for an illustration of this elimination).
\end{proof}

The following theorem summarizes the results of this section.

\begin{theorem}
\label{thm:inverse correspondence theorem}
For every Kracht $\mathcal{L}^\mathrm{FO}_{spd}$-formula $\chi$  and every spd-algebra $\mathbb{H}$,
 \[\mathbb{H}\models \chi \quad \text{ iff }\quad  \mathbb{H}^\ast\models \mathsf{Kracht}(\chi), \]
where $\mathsf{Kracht}(\chi)$ denotes the output of Algorithm \ref{algo:inverse}, i.e.~a conjunction of clopen-analytic $\mathcal{L}_\mathrm{LE}$-inequalities  which can effectively be computed using only (inverse) $\mathsf{ALBA}$-steps (cf.~Section \ref{sec:correspondence}) as described in the procedure above.
\end{theorem}

\begin{remark}
When applied to the {\em general} (i.e.~not necessarily distributive) lattice signature in which $\mathcal{F}_A = \varnothing = \mathcal{G}_A$, the procedure above  allows to identify  the shape of those Kracht formulas corresponding to analytic but not clopen-analytic inequalities, which is the one in which all terms $s, t$ of the atoms $sRt$ in $\eta$ and in $\zeta$  reduce to  variables:   indeed,  by items  7 and 9 of the definition above, each such $s$ and $t$   has shape $\bigvee_i u_i$ (resp.~$\bigwedge_i u_i$) if occurring positively (resp.~negatively), and hence modulo splitting on $\rho(sRt)$, each such relational atom can be further reduced to one of the form $xRy$.
\end{remark}

\begin{algorithm}
\caption{Display node}
\label{algo:display}
\small
\hspace*{\algorithmicindent} \textbf{Input}: An inequality $\varphi \leq \psi$ encoded by its syntax tree, and a displayable node $d$ in the tree, to be  displayed. Nodes in the tree have a sequence of children (in the field ``children'') and a label (in the field ``label'').\\
\hspace*{\algorithmicindent} \textbf{Output}: No output is returned, we assume $\varphi\leq\psi$ to be passed by reference, i.e., any change performed by the function to the inequality is reflected in the caller.
\begin{algorithmic}[1]
\Function{Display}{$\varphi\leq\psi, d$}
    \If{$d$ is the root of either $\varphi$ or $\psi$}
        \State\Return
    \EndIf
    \State{\textbf{let} $p$ be the parent node of $d$, $h$ the connective labelling it, and $c_1,\ldots,c_n$ its children}
    \If{$p \neq \mathrm{root}(\varphi)$ and $p \neq \mathrm{root}(\psi)$}
        \State{\textsc{Display}$(\varphi\leq\psi, p)$}
    \EndIf
    \State{\textbf{let} $i$ be the index s.t.\ $c_i = d$, and $h^{\mathrm{res}}_i$ the residual of $h$ w.r.t.\ the $i$-th coordinate}
    \If{($d$ lies in $\varphi$ and $\varepsilon_h(i) = 1$) \textbf{or} ($d$ lies in $\psi$ and $\varepsilon_h(i) = \partial$)}
        \State{$\mathrm{root}(\varphi) \leftarrow d$}
        \State{$p.\mathrm{children}[i] \leftarrow \mathrm{root}(\psi)$}
        \State{$p.\mathrm{label} \leftarrow h^\mathrm{res}_i$}
        \State{$\mathrm{root}(\psi) \leftarrow p$}
    \Else
        \State{$\mathrm{root}(\psi) \leftarrow d$}
        \State{$p.\mathrm{children}[i] \leftarrow \mathrm{root}(\varphi)$}
        \State{$p.\mathrm{label} \leftarrow h^\mathrm{res}_i$}
        \State{$\mathrm{root}(\varphi) \leftarrow p$}
    \EndIf
\EndFunction
\end{algorithmic}
\end{algorithm}

\begin{algorithm}
\caption{Compactification}
\label{algo:compactification}
\small
\hspace*{\algorithmicindent} \textbf{Input}: A set of variables $D$ to rewrite in display position. A sequence of restricted quantifiers $Q$, each of which having a field ``restricted'' containing the sequence of restricted variables, a field ``restrictor'' containing the restrictor of the quantifier, and a field ``inequality'' containing the restricting inequality of the quantifier. A set of inequalities $I$ (in practice, it will be $\eta$ or $\zeta$). The set $D$ should include all the variables introduced (restricted) by quantifiers in $Q$.\\
\hspace*{\algorithmicindent} \textbf{Output}:
\begin{algorithmic}[1]
\Function{Compactify}{$D, Q, I$}
    \For{$ineq \in I$}
        \State \textbf{rewrite} $ineq$ by pushing $\wedge$ and $\vee$ up as much as possible using the normality of the connectives
    \EndFor
    \While{$I$ contains some inequality $ineq$ of the form $\varphi \leq \psi_1 \wedge \psi_2$ or $\varphi_1 \vee \varphi_2 \leq \psi$}
        \State \textbf{delete} $ineq$ in $I$
        \State $I \leftarrow I \cup \{ \varphi \leq \psi_1, \varphi \leq \psi_2 \}$ \textbf{if} $ineq$ has shape $\varphi \leq \psi_1 \wedge \psi_2$ \\
    \hspace*{\algorithmicindent}\hspace*{\algorithmicindent}\phantom{$I \leftarrow$} $I \cup \{\varphi_1 \leq \psi, \varphi_2 \leq \psi \}$ \textbf{otherwise}
    \EndWhile
    \For{every $ineq \in I$ containing an occurrence $x$ of a variable in $D$}
        \State{\textsc{Display}$(ineq, x)$}
    \EndFor
    \For{every variable $x$ restricted in some quantifier in $Q$ and occurring also in $I$}
        \State\textbf{let} $x\_occ$ be the set of inequalities in $I$ where $x$ occurs (displayed), and let $\alpha_1,\ldots,\alpha_n$ be the terms on the other side in these inequalities
        \State \textbf{let} $new\_ineq \leftarrow x \leq \alpha_1 \wedge \ldots \wedge \alpha_n$ \textbf{if} $x$ occurs displayed on the left in the inequalities in $x\_occ$ \\\hspace*{\algorithmicindent}\hspace*{\algorithmicindent}\phantom{\textbf{let} $new\_ineq \leftarrow$ } $\alpha_1 \vee \ldots \vee \alpha_n \leq x$ \textbf{otherwise}
        \State $I \leftarrow I \setminus \{x\_occ\} \cup \{new\_ineq\}$
    \EndFor
    \For{all quantifiers in $q$ in $Q$ in reverse order}
        \If{no variable in $q.\mathrm{restricted}$ has the same polarity in $I$ and $q.\mathrm{inequality}$}
            \State{\textbf{let} $new \leftarrow $ the inequality built from the inequalities containing (displayed) variables in $q.\mathrm{restricted}$ and $q.\mathrm{inequality}$ as shown in Table \ref{table:deflattification} }
            \If{$q.\mathrm{restrictor}$ occurs displayed on some inequality $ineq$ in $I$}
                \State $new \leftarrow $ the result of inverse splitting between $new$ and $ineq$
                \State \textbf{remove} $ineq$ from $I$
            \EndIf
            \State $I \leftarrow I \cup \{new\}$
        \EndIf
        \State\textbf{remove} all inequalities containing (displayed) variables in $q.\mathrm{restricted}$ from $I$
        \State \textbf{remove} $q$ from $Q$
    \EndFor
\EndFunction
\end{algorithmic}
\hspace*{\algorithmicindent} {\footnotesize{\bf Note}: This algorithm implements the ideas in Proposition \ref{prop:flattify_consequent} and Proposition \ref{prop:flattify_antecedent}. Lines 2-12 implement the ideas in Lemma \ref{lemma:atoms_in_zeta_preprocessed} and \ref{lemma:atoms_in_eta_preprocessed}, lines 13-18 apply inverse splitting to group together all inequalities sharing a displayed variable, and the loop in lines 19-30 implements the procedure described in Proposition \ref{prop:flattify_consequent}. The second case described in the procedure is found in lines 20-28, while the first case is practically reduced to the execution of only line 28.}
\end{algorithm}

\begin{algorithm}
\caption{\textsf{Kracht}}
\label{algo:inverse}
\small
\hspace*{\algorithmicindent} \textbf{Input}: The sequence of $v$-variables $\overline v$, of $a$-variables $\overline a$, of $b$-variables $\overline b$, of $c$-variables $\overline c$, and of $d$-variables $\overline d$. A set of existential restricted quantifiers $Q_c$ restricting variables in $\overline c$, and a set of universal restricted quantifiers $Q_d$ restricting variables in $\overline d$. Two sets of inequalities $\eta$ and $\zeta$ satisfying Definition \ref{def:inverse_shape}. \\
\hspace*{\algorithmicindent} \textbf{Output}: A set of inequalities the conjunction of which corresponds to the input formula.
\begin{algorithmic}[1]
\Function{Kracht}{$\overline v, \overline a, \overline b, \overline c, \overline d, Q_c, Q_d, \eta, \zeta$}
    \State{\textsc{Compactify}$(\overline d, Q_d, \zeta)$}
    \State{\textsc{Compactify}$(\overline a \oplus \overline b \oplus \overline c, Q_c, \eta)$}
    \For{every inequality $ineq$ with a displayed variable $x$ in $\overline a$ or $\overline b$ in $\eta$}
        \State{\textbf{let} $\alpha$ be the term on the other side of the inequality w.r.t.\ $x$}
        \If{$x$ occurs in $\zeta$}
            \State{\textbf{substitute} $\alpha$ to the only occurrence of $x$ in $\zeta$}
        \EndIf
    \EndFor
    \State\Return $\zeta$
\EndFunction
\end{algorithmic}
\hspace*{\algorithmicindent} {\footnotesize \textbf{Note}: Line 2 applies Lemma \ref{prop:flattify_consequent}, line 3 applies Lemma \ref{prop:flattify_antecedent}, and lines 4-9 apply the inverse of the first approximation, i.e.\ Proposition \ref{prop:inverse_first_approx}.}
\end{algorithm}

\section{Dual characterizations along discrete duality for distributive spd-algebras}
\label{sec: applications}
In the present section, we firstly identify {\em inductive Kracht formulas} as those Kracht formulas $\chi$  such that each conjunct of $\mathsf{Kracht}(\chi)$ is an analytic {\em inductive} $\mathcal{L}$-inequality (cf.~Lemma \ref{lemma:kracht_inductive_to_inductive}); then, as a corollary of the inverse correspondence result of the previous section, we extend the dual characterizations of finitely many conditions on subordination lattices in terms of conditions on subordination spaces (cf.~\cite[Proposition 6.11 and 6.12]{de2024obligations2}) to   inductive Kracht formulas. This generalization  follows from the {\em standard} correspondence results for inductive (D)LE-inequalities between perfect algebras and their associated frames \cite{conradie2012algorithmic,conradie2019algorithmic}, using Theorem \ref{thm:inverse correspondence theorem}, and the slanted canonicity of analytic inductive inequalities \cite[Theorem 4.1]{de2021slanted}.

An analytic $\mathcal{L}$-inequality $\phi\leq\psi$ is {\em inductive} (cf.~\cite[Definition 3.4]{conradie2019algorithmic}, \cite[Definition 18]{greco2018unified}) if some order-type $\varepsilon$ and some irreflexive and transitive relation $\Omega$ exist on the set of variables occurring in $\phi\leq\psi$ such that
every SRR node $h$ lies on at most one $\varepsilon$-critical branch in the signed generation tree of $\phi\leq\psi$ (where, for any variable $v\in \mathsf{Prop}$, a branch ending in an occurrence of $v$ is $\varepsilon$-{\em critical}  iff its leaf is positive and $\varepsilon(v) = 1$ or its leaf is negative and $\varepsilon(v) = \partial$), and $u <_\Omega v$ for every variable $u$ labelling each leaf $l'\neq l$ in the scope of $h$.


\begin{definition}
\label{def:kracht_analytic_inductive}
A Kracht $\mathcal{L}^\mathrm{FO}_{spd}$-formula (cf.~Definition \ref{def:inverse_shape}) is {\em inductive} if:
\begin{itemize}
\item for every atom $sRt$ in $\zeta$, every node in the signed generation tree of $\rho(sRt)$ is a skeleton node;
\item there is an order type $\varepsilon$ and a strict order $\Omega$ on the variables in $\overline v$ such that, for every $sRt$ in $\eta$, every $\varepsilon$-critical occurrence $l$ in the dual signed generation tree of $\rho(sRt)$ (cf.~Section \ref{ssec:analytic-LE-axioms}) is such that
    \begin{itemize}
        \item the branch to $l$ in $\rho(sRt)$ is good;
        \item for every SRR node $\oast(s_1,\ldots,s_n)$ in the branch to $l$ in $\rho(sRt)$, if $l$ lies in $s_i$, each $s_j$ with $j\neq i$ does not contain any $\varepsilon$-critical variable, and any variable $u$ occurring in $s_j$ is such that $u <_\Omega v$, where $v$ is the variable labelling $l$.
    \end{itemize}
\end{itemize}
\end{definition}

\begin{lemma}
\label{lemma:kracht_inductive_to_inductive}
For every inductive Kracht $\mathcal{L}_{spd}^\mathrm{FO}$-formula $\chi$, every $\mathcal{L}_{\mathrm{LE}}$-inequality in $\mathsf{Kracht}(\chi)$ is analytic inductive.
\end{lemma}
\begin{proof}
The only connectives introduced during the application of the inverse correspondence algorithm of Section \ref{sec:inverse_correspondence_procedure} are the lattice connectives introduced by inverse splitting rules (cf.\ Footnote \ref{footnote:inverse_splitting}), and the slanted connectives introduced by translating relational atoms  in $\eta$ into flat inequalities. These  are all SRA nodes, and hence they  can never introduce violations of the defining condition of analytic inductive inequalities between the variables in $\overline v$. Hence, since  each $\mathcal{L}_{\mathrm{LE}}$-inequality in $\mathsf{Kracht}(\chi)$ is already guaranteed to be analytic (except for the clopen terms at the bottom of the syntax tree of the inequality), the only violations  can be found in the clopen terms, since they can contain branches which are not good, and SRR nodes. Definition \ref{def:kracht_analytic_inductive} introduces precisely the required restrictions which ensure that these clopen terms be well behaved.
\end{proof}

\begin{definition} (cf.~\cite[Definitions 6.8 and 6.10]{de2024obligations2})
 \label{def:spd-space }
 The  {\em spd-space}\footnote{We warn the reader that, for the sake of a more direct application of \cite[Theorem 4.1]{conradie2019algorithmic} in the proof of Proposition \ref{prop:dual characterization}, we define $R_{\prec}$, $R_{\npcon}$ and $R_{\nppcon}$ as the {\em converses} of the relations defined in \cite[Definitions 6.8 and 6.10]{de2024obligations2}. Of course there is no conceptual difference between the two definitions; it is just a matter of convenience of presentation.} associated with an spd-algebra $\mathbb{H}$ based on a distributive lattice $A$ is $\mathbb{H}_*: = ((\jty(A^\delta), \sqsubseteq), \mathcal{R}_\mathsf{S}, \mathcal{R}_\mathsf{C},\mathcal{R}_\mathsf{D})$, where
 \begin{enumerate}
 \item $(\jty(A^\delta), \sqsubseteq)$ is the poset of the completely join-irreducible elements\footnote{An element $j\neq \bot$ of a (complete) lattice $L$ is {\em completely join-irreducible} (resp.~{\em completely join-prime}) if $j\in S$ (resp.~$j\leq s\in S$) whenever $j = \bigvee S$ (resp.~$j\leq \bigvee S$) for  $S\subseteq L$. Clearly, any completely join-prime element is completely join-irreducible, and the two notions are equivalent if, in $L$, finite meets distribute over arbitrary joins.} of $A^\delta$ s.t.~$j\sqsubseteq k$ iff $j\geq k$;
 \item $\mathcal{R}_\mathsf{S} = \{R_{\prec}\mid {\prec}\in \mathsf{S}\}$ where $R_\prec$ is a binary relation on $ \jty(A^\delta)$ for any ${\prec}\in \mathsf{S}$, such that $  R_{\prec}(i, j)$ iff $i\leq \Diamond j$ iff $\blacksquare \kappa (i)\leq \kappa(j)$, where $\kappa (i): = \bigvee\{i'\in \jty(A^\delta)\mid i\nleq i'\}$;\footnote{\label{footn:kappa} From the definition of $\kappa$ and the join primeness of every $j\in J^{\infty}(A^\delta)$, it immediately follows that $j\nleq u$ iff $u\leq \kappa(j)$ for every $u \in A^\delta$ and  every $j \in J^{\infty}(A^\delta)$.}
 \item $\mathcal{R}_\mathsf{C} = \{R_{\npcon}\mid \pcon\in \mathsf{C}\}$ where $R_{\npcon}$ is a binary relation on $ \jty(A^\delta)$ for any $\pcon\in \mathsf{C}$, such that  $  R_{\npcon}(i, j)$  iff $\wt  j\leq \kappa(i)$ iff  $\bt  i\leq \kappa(j)$.

  \item $\mathcal{R}_\mathsf{D} = \{R_{\ppcon}\mid \ppcon\in \mathsf{D}\}$ where $R_{\nppcon}$ is a binary relation on $ \jty(A^\delta)$ for any $\ppcon\in \mathsf{D}$,   such that $  R_{\nppcon}(i, j)$  iff $i\leq \tw \kappa(j)$ iff  $ j\leq \tb \kappa (i)$.
 \end{enumerate}
\end{definition}
From the definition above, it immediately follows that
\begin{lemma}
\label{lemma:spd-space is discrete dual of canext}
    For any spd-algebra $\mathbb{H}$ based on a distributive lattice $A$, the spd-space $\mathbb{H}_\ast$ is isomorphic to the prime structure $(\mathbb{H}^\delta)_{+}$ (cf.~\cite[Definition 2.12]{conradie2012algorithmic}) of the perfect distributive modal algebra $\mathbb{H}^\delta$.
\end{lemma}

\begin{proposition}
\label{prop:dual characterization}
    For any  spd-algebra $\mathbb{H}$ based on a distributive lattice,
    and any inductive Kracht formula $\chi$,
    \[\mathbb{H}\models \chi \quad \text{ iff }\quad  \mathbb{H}_\ast\models \mathsf{ALBA}(\mathsf{Kracht}(\chi)), \]
 where $\mathsf{ALBA}(\mathsf{Kracht}(\chi))$ denotes the conjunction of the first-order correspondents on spd-spaces of the analytic inductive $\mathcal{L}_\mathrm{LE}$-inequalities in    $\mathsf{Kracht}(\chi)$.
\end{proposition}
\begin{proof}
    By Theorem \ref{thm:inverse correspondence theorem}, $\mathbb{H}\models \chi$ iff $  \mathbb{H}^\ast\models \mathsf{Kracht}(\chi)$; since $\chi$ is an inductive Kracht formula, by Lemma \ref{lemma:kracht_inductive_to_inductive}, each $\mathcal{L}_\mathrm{LE}$-inequality in $\mathsf{Kracht}(\chi)$ is analytic inductive, and hence {\em slanted} canonical (cf.~\cite[Theorem 4.1]{de2021slanted}). Therefore, $  \mathbb{H}^\ast\models \mathsf{Kracht}(\chi)$ iff $  \mathbb{H}^\delta\models \mathsf{Kracht}(\chi)$, and since $\mathbb{H}^\delta$ is a standard distributive modal algebra (cf.~Definition \ref{def: sigma and pi extensions of slanted}), by \cite[Theorem 4.1]{conradie2019algorithmic},  $  \mathbb{H}^\delta\models \mathsf{Kracht}(\chi)$ iff $  (\mathbb{H}^\delta)_+\models \mathsf{ALBA}(\mathsf{Kracht}(\chi))$. The required  statement is then an immediate consequence of Lemma \ref{lemma:spd-space is discrete dual of canext}.
\end{proof}

\section{Conclusions}
\label{sec:conclusions}

In the present paper, pivoting on the notion of slanted algebras \cite{de2021slanted}, we have extended insights, methods, techniques, and results of {\em unified  correspondence} theory
\cite{conradie2014unified, conradie2022unified, palmigiano2024unified}  to input output/logic and subordination algebras and related structures, two  research areas  which have been  connected with each other only very recently \cite{wollic22, de2024obligations, de2024obligations2}. The `unified correspondence' methodology is grounded  on the recognition of correspondence phenomena as {\em dual characterizations}, i.e.~characterizations across suitable dualities between algebras and relational structures. This insight has allowed for correspondence-theoretic phenomena  to be studied and explained in terms of the order-theoretic properties of the algebraic interpretations of the logical connectives. Via dual characterization, it was possible for core correspondence-theoretic techniques to be systematically transferred    from relational structures to {\em algebras}, and leveraging on  this vantage point, for Sahlqvist-type  results  to be  significantly generalized and extended   from modal logic to nonclassical logics \cite{conradie2012algorithmic,conradie2017algebraic}. Unlike previous  approaches to correspondence theoretic-results in the context of subordination algebras \cite{celani2020subordination, de2020subordination}, in the present paper, subordination algebras and related structures (in particular {\em spd-algebras}, cf.~Definition \ref{def:spd-algebra}) play the same role played by {\em relational structures} in classsical correspondence theory.\footnote{Not in the sense that spd-algebras and slanted algebras are dual to each other; they are in fact  {\em equivalent} to each other, and via this equivalence, a correspondence can be established between formulas of the first-order language $\mathcal{L}^{\mathrm{FO}}_{spd}$ and $\mathcal{L}_{\mathrm{LE}}$-inequalities.} As such, the frame correspondence language (cf.~Definition \ref{def:fo-language-spd}) associated with spd-algebras admits not only binary relation symbols but also function symbols, corresponding to the algebraic signature of the domain of  spd-algebras. This language naturally encodes (the algebraic counterparts of) closure rules considered and studied in input/output logic. The {\em algebraic} side of the present correspondence-theoretic setting is given by slanted algebras, i.e.~algebraic structures based on lattice expansions $A$ endowed with {\em slanted  operations}, i.e.~maps of type $A\to A^\delta$. Thanks to the notion of slanted operation, output operators in input/output logic can be regarded as modal operators in a  mathematically explicit way (cf.~\cite[Section 6.1]{de2024obligations2}). In previous work \cite[Section 5]{de2024obligations2}, we have applied correspondence-theoretic techniques on a case-by-case basis, to modally characterize  a finite number of well known rules and conditions considered both in the  input/output logic literature and in the literature on subordination algebras and related structures; in the present paper, we have extended  these results so as to cover the class of Kracht formulas (cf.~Definition \ref{def:inverse_shape}) on the side of spd-algebras, and analytic  inequalities (cf.~Section \ref{ssec:analytic-LE-axioms}) on the side of propositional logics and their associated slanted algebras. These results have made it possible to generalize and extend previous correspondence results for subordination algebras \cite{celani2020subordination, de2020subordination}, while at the same time clarifying their connections to standard correspondence results (cf.~Proposition \ref{prop:dual characterization}). The results of the present paper pave the way to a number of further research directions, discussed below.

\paragraph{Dual characterizations beyond the distributive setting} As mentioned above, the literature on subordination algebras and related structures has considered correspondence-theoretic results in the setting of Boolean algebras and distributive lattices, in which, rather than being characterized as modal axioms as is done in the present paper, first-order conditions on subordination algebras are characterized  as {\em first-order conditions} on relational structures (called {\em subordination spaces}  cf.~\cite[Definition 4.5]{celani2020subordination}). The generalization of these results developed in Section \ref{sec: applications} is also presented in a distributive setting. However, the techniques and results on which this generalization hinges hold even more broadly; specifically, they hold in the setting of LE-algebras in which distributive laws do not need to hold. A natural generalization of the results developed in Section \ref{sec: applications} concerns the general lattice setting, in which the counterparts of spd-spaces can  be defined as suitable classes of polarity-based frames or graph-based frames \cite{stolpe2015concept,conradie2020non, conradie2024modal, conradie2024modal2}.

\paragraph{Normative systems and slanted algebras with binary operators} The standard framework in input/output logic concerns normative and permission systems represented as binary relations. Their algebraic incarnations naturally corresponds to the binary relations of spd-algebras, which in turn give rise to unary modal operators on slanted algebras.  However, slanted algebras admit slanted operators of any arity, so a natural direction in which the present results can be expanded concerns  the algebraic theory of normative and permission systems of higher arity, which have already been considered and studied in the literature \cite{Pigozzi18}.
As to the subordination algebra literature, it is well known that subordination algebras can equivalently be represented as Boolean algebras or distributive lattices  $A$ with {\em quasi-modal} operators (cf.~\cite{celani2001quasi}), i.e.~unary maps of type $A\to Filt(A)$ or $A\to Id(A)$, where $Filt(A)$ and $ Id(A)$ respectively denote the sets of the filters and ideals of $A$. In view of  the well-known isomorphisms $Filt(A)\cong K(A^\delta)$ and $Id(A)\cong O(A^\delta)$ (cf.~\cite{Gehrke2004BoundedDL}), quasi-modal operators can equivalently be represented as slanted operators of the same arity (cf.~\cite[Section 6.4]{de2024obligations2}). Algebras with {\em binary} quasi-modal operators have been introduced in \cite{castro2003generalized,castro2011distributive}, and recently, they have been studied in connection with subordination relations \cite{calomino2023study}. We are presently  exploring this connection, also in view of potential applications to the theory of input/output logics.

\paragraph{Modal embeddings and proof calculi  for input/output logics} Since its inception \cite{Makinson00}, the literature on input/output logic has striven to devise  syntactic and computational characterizations of the sets of obligations and compatible permissions logically deducible from a given set of conditional norms; moreover, since its inception,  establishing explicit connections between input/output frameworks and modal logic frameworks, e.g.~via embeddings and translations,  has been regarded as a  possible route for achieving this aim. However, modal embeddings into deontic logic have proven to be  highly dependent on the specific definitions of certain outputs, and hence not uniformly available, highly nonmodular, and technically challenging  (cf.~\cite{strasser2016adaptive}). Very recently, modular sequent calculi for input/output  logics have been introduced \cite{ciabattoni2022dyadic} for  the four original outputs of \cite{Makinson00} and of their causal counterparts introduced in \cite{bochman2004causal}, as well as helping to establish complexity bounds, and define possible worlds semantics
and formal embeddings into normal modal logics, leading to formally explicit modal characterizations of the output operators mentioned above.  The modal characterization results in \cite{ciabattoni2022dyadic}   are very similar to those of Section \ref{sec:correspondence}, which hold for a class of modal axioms encompassing the {\em analytic inductive} axioms (\cite[Definition 55]{greco2018unified}, see also Section \ref{sec: applications}). In \cite{greco2018unified}, analytic inductive axioms have been characterized as those axioms which can be equivalently transformed into rules which preserve cut elimination when added to a calculus with cut elimination. We are presently exploring the possibility of applying the results of the present paper in the modular design of proof calculi with good computational properties for classes of input/output logics. The availability of this uniform proof theoretic framework might be particularly useful vis-\`a-vis the recent convergence between input/output logic and formal argumentation theory \cite{ berkel2022reasoning, chen2024bisimulation}.

\paragraph{Implementations and provers}  The ongoing work on implementing input/output logic using LogiKEy’s methodologies is grounded in our recent efforts described in \cite{de2024obligations2}. It is straightforward to implement the corresponding modal algebras of obligation, negative permission, and dual-negative permission aligned with original proposed syetems.  The soundness result demonstrate that the efficiency of the algebraic encoding is  similar to the LogiKEy benchmark examples~\cite{Christophdata}.\footnote{See the GitHub link (direct implementation): \href{ https://github.com/farjami110/AlgebriacInputOutput}{ https://github.com/farjami110/AlgebriacInputOutput}} However, the more general results presented in the present paper offer a broader and more principled platform for these implementations. By establishing a systematic connection between input/output logic and various modal algebras, our approach will allow for more efficient and accurate use of off-the-shelf theorem provers in the development of responsible AI systems, enabling the LogiKEy framework to address a wider array of logical formalisms with greater computational efficiency.

\section*{Acknowledgments}

\subsection*{Competing interests}
    The authors of this study declare that there is no conflict of interest with any commercial or financial entities related to this research.
\subsection*{Authors' contributions}
    Xiaolong Wang drafted the initial version of this article. Other authors have all made equivalent contributions to it.
\subsection*{Funding}
    The authors  affiliated to the Vrije Universiteit Amsterdam have received funding from the EU Horizon 2020 research and innovation programme under the MSCA-RISE grant agreement No.~101007627.
\\
Xiaolong Wang is supported by the China Scholarship Council No.~202006220087.
\\
Krishna Manoorkar is supported by the NWO grant KIVI.2019.001 awarded to Alessandra Palmigiano.

\bibliographystyle{abbrv}

\end{document}